\documentclass[12pt,psfig,a4]{article}
\makeatletter
\def\@seccntformat#1{\@ifundefined{#1@cntformat}%
   {\csname the#1\endcsname\quad}  % default
   {\csname #1@cntformat\endcsname}% enable individual control
}
\let\oldappendix\appendix %% save current definition of \appendix
\renewcommand\appendix{%
    \oldappendix
    \newcommand{\section@cntformat}{\appendixname~\thesection:\,\,}
}
\makeatother

\usepackage{amsfonts}
\usepackage{geometry}
\usepackage{graphics}
\usepackage{lscape}
\usepackage{subfigure}
\usepackage{graphicx}
\usepackage{setspace}
\usepackage{amsthm}
\usepackage{amssymb}
\usepackage{amsmath}
\usepackage{color,soul}
\usepackage{epstopdf}
\usepackage{multirow}
\usepackage{graphicx}
\usepackage{hhline}
\usepackage{enumitem}
\usepackage[toc,page]{appendix}
\usepackage{bbm}

\usepackage{booktabs}
\usepackage{xcolor}
\usepackage{colortbl}

\usepackage{xcolor}
\usepackage[round]{natbib}
\PassOptionsToPackage{
        natbib=true,
        style=authoryear-comp,
        hyperref=true,
        backend=biber,
        maxbibnames=99,
        firstinits=true,
        uniquename=init,
        maxcitenames=2,
        parentracker=true,
        url=false,
        doi=false,
        isbn=false,
        eprint=false,
        backref=true,
            }   {biblatex}

\newcommand{\Lyx}{L\kern-.1667em\lower.25em\hbox{y}\kern-.125emX\spacefactor1000}
\newcommand{\rev}[1]{\textcolor{blue}{#1}}

\newcommand{\bm}[1]{\boldsymbol{#1}}

\newtheorem{theorem}{Theorem}

\newtheorem{example}[theorem]{Example}

\newtheorem{lemma}[theorem]{Lemma}

\newtheorem{proposition}[theorem]{Proposition}
\newtheorem{remark}[theorem]{Remark}

\newtheorem{assumption}[theorem]{Assumption}
\renewcommand{\O}{\Omega}
\renewcommand{\o}{\omega}

\newcommand{\F}{\mathcal{F}}

\newcommand{\E}{\mathbb{E}}
\newcommand{\N}{\mathbb{N}}

\newcommand{\R}{\mathbb{R}}
\renewcommand{\P}{\mathbb{P}}

\let\abs=\envert

\DeclareMathOperator{\Cov}{Cov}
\DeclareMathOperator{\V}{Var}
\DeclareMathOperator{\MAD}{MAD}
\DeclareMathOperator{\MSAD}{MSAD}
\DeclareMathOperator{\sgn}{sgn}
\DeclareMathOperator{\argmin}{argmin\ }
\DeclareMathOperator{\RC}{RC}

\newcommand{\of}[1]{\ensuremath{\left( #1 \right)}}
\newcommand{\cb}[1]{\ensuremath{ \left\{ #1 \right\} }}
\newcommand{\sqb}[1]{\ensuremath{ \left[ #1 \right] }}

\def\prehp(#1,#2){\ensuremath{  #1 \cdot #2 }}

\definecolor{armygreen}{rgb}{0.19, 0.53, 0.43}

\input{tcilatex}

\begin{document}

\title{
\textbf{MAD Risk Parity Portfolios}}
\author{\c{C}a\u{g}{\i}n Ararat$^1$, Francesco Cesarone$^2$, Mustafa \c{C}elebi P{\i}nar$^1$, Jacopo Maria Ricci$^{2,3}$\\
	{\small $^1$ \emph{Bilkent University - Department of Industrial Engineering}}\\
	{\footnotesize cararat@bilkent.edu.tr, mustafap@bilkent.edu.tr} \\
	{\small $^2$\emph{Roma Tre University - Department of Business Studies}}\\
	{\footnotesize francesco.cesarone@uniroma3.it, jacopomaria.ricci@uniroma3.it}\\
	{\small $^3$\emph{Bergamo University - Department of Economics}}\\
	{\footnotesize jacopomaria.ricci@unibg.it}\\
}
%\date{\today}
\maketitle

%\begin{abstract}
%	
%	
%	\bigskip \noindent \textbf{Keywords}:
%\end{abstract}

\begin{abstract}

In this paper, we investigate the features and the performance of the Risk Parity (RP) portfolios using the Mean Absolute Deviation (MAD) as a risk measure.
The RP model is a recent strategy for asset allocation that
aims at equally sharing the global portfolio risk among all the assets of an investment universe. We discuss here some existing and new results about the properties of MAD that are useful for the RP approach.
We propose several formulations for finding MAD-RP portfolios computationally, and compare them in terms of accuracy and efficiency.
Furthermore, we provide extensive empirical analysis based on three real-world datasets, showing that the performances of the RP approaches generally tend to place both in terms of risk and profitability between those obtained from the minimum risk and the Equally Weighted strategies.

\end{abstract}

\section{Introduction}

The 2008 subprime financial crisis led many scholars and practitioners
to strong criticism of classical risk-gain models and to the consequent development of new portfolio selection strategies
based on the concept of risk allocation.
Indeed,
even though risk-gain models often have nice features in terms of formulation and of tractability,
they often show several drawbacks, such as their high sensitivity to estimation errors of the input parameters \citep[in particular, of expected returns, see, e.g.,][]{Best1991b,Best1991a,Chopra1993,Michaud1998,Demiguel2009},
or their lack of risk diversification.
%
%The second issue is due to the fact that these portfolio optimization models usually select few assets; consequently, the resulting portfolio might not be balanced in terms of risk attributed to each asset.
%
A straightforward method to deal with this issue could be
the choice of the Equally Weighted (EW) portfolio,
where the invested capital is equally distributed
among the assets that belong to the investment universe.
However, if the investment universe consists of assets with very different
intrinsic risks, then the resulting portfolio has limited total risk diversification.
%
%Consequently, more refined approaches have been developed, for example the Risk Parity approach (see ). The authors build a portfolio in which each asset has an equal risk contribution to the whole portfolio risk; in this specific case, the risk is measured by the volatility.
%

A more refined risk-focused method is the Risk Parity (RP), also called Equal Risk Contribution. It consists of selecting a portfolio where each asset equally contributes to the total portfolio
risk, regardless of their estimated expected returns \citep{Maillard2010}.
The RP strategy has its roots
in the practitioner world \citep[see, e.g.,][]{fabozzi2021risk,liu2020}, indeed it is often considered as a heuristic method.
However,
several recent theoretical and empirical findings justify the growing popularity of the RP strategy
in the practice of asset allocation.

From an empirical viewpoint,
Risk Parity portfolios show smaller sensitivity to estimation errors of the input parameters
compared to portfolios based on risk minimization and on other optimization strategies \citep{cesarone2020stability}.
Furthermore, Risk Parity portfolios seem to show promising out-of-sample performance
\citep{fisher2015risk,jacobsen2020risk}.

From a theoretical viewpoint,
since an RP portfolio
can be found by solving a minimization problem with logarithmic constraints on the
portfolio weights,
the RP strategy can be interpreted as a minimum risk approach with a constraint on
the minimum level of diversification, which can be seen as a sort of regularization.
Furthermore,
an example of the theoretical justification of the Risk Parity approach can be found in
\cite{oderda2015stochastic},
where the author
shows that
the analytic form of the optimal
portfolio solution
obtained by maximizing
portfolio relative logarithmic
wealth at a fixed tracking risk level
with respect to the market-capitalization-weighted index,
consists of a linear combination
of this index,
the global minimum
variance portfolio, the EW portfolio, the RP
portfolio, and the high cash flow rate of return portfolio.

The risk measure commonly used in the RP approach is volatility \citep{Maillard2010,Roncalli2013}. %
Other risk
measures have been considered in the literature
such as the Conditional Value-at-Risk \citep{Boudt2013,cesarone2018minimum,mausser2018long} and the Expectiles \citep{bellini2021risk}, both belonging to the class of coherent risk measures.
However, the long-only Risk Parity model
has a unique solution when the risk is positive, convex
and positively homogeneous \citep{Cesarone2019Anoptimization}, and
for Conditional Value-at-Risk and Expectiles,
positivity is not always guaranteed \citep{cesarone2017equal,cesarone2018minimum,bellini2021risk}.	
An alternative risk measure to volatility, which is, by definition, positive for nonconstant market returns,
is the Mean Absolute Deviation (MAD) belonging to the
class of deviation risk measures \citep{Rockafellar2006}.
In risk-gain portfolio optimization analysis, it was introduced
by \cite{Konno91} as an alternative to the Markowitz model.

In this paper, the main goal is to investigate the theoretical properties and the performance of
the MAD Risk Parity portfolios,
providing, to the best of our knowledge, multiple contributions to the literature. First, we revisit existing theoretical results on MAD,
discuss its differentiability, and provide a characterization of multivariate random market returns for which MAD is additive.
We show the conditions that determine this characterization,
which can be seen as strong positive dependence among the asset returns.
Furthermore, under these conditions,
we provide a closed-form solution for the long-only MAD-RP portfolio.
Second, we establish the existence and uniqueness
of the MAD-RP portfolio, thus extending the theoretical results of \cite{Cesarone2019Anoptimization}.
Third, we propose several formulations to find the MAD-RP portfolios practically, thus comparing their performances both in terms of accuracy and efficiency.
Finally, we provide an extensive empirical analysis on three real-world datasets by
comparing the out-of-sample performance obtained from the global minimum volatility and MAD portfolios, the volatility and MAD Risk Parity portfolios, and the Equally Weighted portfolio.

The rest of the paper is organized as follows.
Section \ref{MAD} introduces the Mean Absolute Deviation risk measure,
providing a discussion of MAD properties relevant to our context.
In Section \ref{sec:RPwithMAD}, we show how to formulate the RP approach with MAD mathematically and how to find the MAD-RP portfolios in practice.
Section \ref{sec:EmpirCompRes} provides an extensive empirical analysis based on three real-world datasets.
More precisely, in Section \ref{sec:AccuEffi}, we test and compare all the MAD-RP formulations in terms of accuracy and efficiency, while in Sections \ref{sec:insample} and \ref{sec:out-of-sample}, we report a graphical comparison of some portfolio selection approaches and out-of-sample results, respectively. Finally, in Section \ref{sec:ConclusionsMADRP}, we draw some overall conclusions.

\section{Measuring the portfolio risk by MAD}\label{MAD}

Let $n\in\N$. We denote by $\R^n$ the $n$-dimensional Euclidean space and define the cones $\R^n_+:=\{\boldsymbol{x}\in\R^n\mid x_i\geq 0\mbox{ for each }i\in\{1,\ldots,n\}\}$, $\R^n_{++}:=\{\boldsymbol{x}\in\R^n\mid x_i>0\mbox{ for each } i\in\{1,\ldots,n\}\}$. We also denote by $\Delta^{n-1}$ the set of all $\boldsymbol{x}\in \R^n_+$ with $\sum_{i=1}^n x_i=1$.

To introduce the probabilistic setup, let us fix a complete probability space $(\O, \F, \P)$. For an event $A\in \F$, we define its indicator function by $\mathbbm{1}_A(\omega)=1$ for $\omega\in A$, and by $\mathbbm{1}_A(\omega)=0$ for $\omega\in \Omega\setminus A$. We denote by $L^0:=L^0(\O,\F,\P)$ the space of all $\F$-measurable and real-valued random variables, where two elements are distinguished up to $\P$-almost sure (a.s.) equality. We denote by $L^1:=L^1(\O,\F,\P)$ the set of all $X\in L^0$ with $\E[\abs{X}]<+\infty$.
%For each $p\in[1,+\infty)$, we denote by $L^p:=L^p(\O,\F,\P)$ the set of all $X\in L^0$ with $\E[\abs{X}^p]<+\infty$.

%All random variables that we consider are defined on the sample space $\Omega$. Furthermore, we denote by $L^{p}:=L^{p}(\Omega, \mathcal{F}, \mathbb{P})$ the space of random variables defined on $(\Omega, \mathcal{F}, \mathbb{P})$ having finite the moment of order $p$.

The \emph{mean absolute deviation} (MAD) of a random variable is a measure of its dispersion and is defined as the expected value of the absolute deviations
from a reference value that is generally represented by the mean of the random variable: %\out{(but it could also be its median, mode, or any other location measure):}
\begin{equation}\label{eq:MADdef}
  \MAD(X):=\E[\abs{X-\E[X]}] \, ,\quad X\in L^1.
\end{equation}
MAD was first used in the field of portfolio optimization by \cite{Konno91} as \rev{a} symmetric risk measure alternative to variance. An upside
%\out{A downside}
version of MAD is the \emph{mean semi-absolute deviation}
%\out{(MSAD)}
defined by
\begin{equation}\label{eq:SMADdef}
  \MSAD(X):=\E[(X-\E[X])^+] \, ,\quad X\in L^1,
\end{equation}
where $a^+:=\max\{a,0\}$ denotes the positive part of $a\in\R$. As it is well-known, we have $\MSAD(X)=\frac12 \MAD(X)$.

\subsection{MAD as a deviation measure \label{subsec:DevMeas}}

The next proposition summarizes the properties of MAD as a \emph{deviation measure}. We omit its straightforward proof for brevity and refer the reader to \cite{Rockafellar2006} for a general treatment of deviation measures.

%Let $\Omega$ be the set of scenarios and $\mathcal{L}^{2}(\Omega)$ be a linear space. Then, the authors define a deviation measure as in below.

%\begin{definition}[Deviation measures]\label{def:devmeas}
\begin{proposition}\label{prop:dev}
The functional $\MAD\colon L^1 \to [0, \infty)$ is a deviation measure, that is, it satisfies the following properties for every $X,Y\in L^1$:
\begin{enumerate}[label=(\roman*)]
	\item \textit{Translation invariance:} $\MAD(X+\alpha)=\MAD(X)$ for every $\alpha\in \R$.
	\item \textit{Positive homogeneity:} $\MAD(\lambda X)=\lambda\MAD(X)$ for every $\lambda\geq 0$.
	\item \textit{Subadditivity:} $\MAD(X+Y)\leq \MAD(X)+\MAD(Y)$.
	\item \textit{Normalization:} $\MAD(0)=0$.
	\item \textit{Strict positivity:} It holds $\MAD(X)>0$ if and only if $X$ is not equal to a constant $\P$-a.s.
\end{enumerate}
\end{proposition}

\subsection{Additivity of MAD}\label{sec:add}

Next, we provide a property that allows for finding a closed form solution to MAD-RP portfolios,
namely, a characterization of random variables for which MAD is additive. To the best of our knowledge, this basic observation is new to the current paper.

\begin{lemma}[Additivity of MAD]\label{MADAdditive}
	Let $X,Y \in L^1$.
Then,
\begin{equation}\label{eq:add}
\MAD(X+Y)=\MAD(X)+\MAD(Y)
\end{equation}
if and only if
	\begin{equation}\label{eq:add3}
	(X-\E[X])(Y-\E[Y])\geq 0 \quad \P\mbox{-a.s}.
	\end{equation}
\end{lemma}

\begin{proof}
  See Appendix \ref{app:Lemma}.
\end{proof}

Just as the portfolio volatility is additive for asset returns that are perfectly correlated
 %(i.e., their linear correlation is equal to 1)
 and the portfolio Conditional Value-at-Risk (CVaR) is additive for asset returns that are comonotone \citep{tasche2002expected}, in view of Lemma~\ref{MADAdditive}, MAD is additive for random returns $X, Y$ if and only if the deviations of $X$ and $Y$ from their respective
means have the same sign almost surely.
Hence, Condition \eqref{eq:add3} can be interpreted as a kind of positive dependence. %In particular, it is necessary (but not sufficient) that $X$ and $Y$ are positively correlated. %
Furthermore, we point out that in the case of additivity of homogeneous and subadditive risk measures, the sum of the risks of $X$ and $Y$ is an upper bound for the risk of $X + Y$. Such an upper bound represents the worst portfolio risk, and it occurs for volatility when the linear correlation between $X$ and $Y$ is equal to 1, for CVaR when $X$ and $Y$ are comonotone, and for MAD when $X$ and $Y$ satisfy Condition \eqref{eq:add3}.

\begin{remark}\label{rem:corr}
	Clearly, if $X,Y$ are square-integrable random variables that satisfy \eqref{eq:add3}, then they are positively correlated (i.e., nonnegative covariance). Moreover, if $X,Y$ are square-integrable random variables with perfect positive correlation, then \eqref{eq:add3} holds. Indeed, having perfect positive correlation implies that $Y=aX+b$ for some $a>0$ and $b\in\R$ so that $(X-\E[X])(Y-\E[Y])=a(X-\E[X])^2\geq 0$ $\P$-a.s. In the next remark and example, we prove that the converse is not true in general.
	\end{remark}

\begin{remark}\label{rem:additivity}
	Let $\alpha\in (0,1)$. For a square-integrable random variable $X$, its $\alpha$-expectile is defined as 
		\[
		e_\alpha(X):=\underset{r\in\R}{\argmin}\E[\alpha((X-r)^+)^2+(1-\alpha)((X-r)^-)^2].
		\]
		In \citet[Theorem~3]{bellini2021risk}, it is shown that the additivity of the $\alpha$-expectiles at two square-integrable random variables $X,Y$, i.e., $e_\alpha(X+Y)=e_\alpha(X)+e_\alpha(Y)$, is equivalent to the condition
		\begin{equation}\label{eq:expectile}
			(X-e_\alpha(X))(Y-e_\alpha(Y))\geq 0\quad \P\mbox{-a.s}.
		\end{equation}
		whenever $\alpha\in (\frac12,1)$. In \citet[Example~2]{bellini2021risk}, it is assumed that $\P\{X=0,Y=0\}=\P\{X=0,Y=1\}=\P\{X=1,Y=0\}=\P\{X=1,Y=1\}=\frac16$ and $\P\{X=2,Y=2\}=\frac13$. In this case, they show that \eqref{eq:expectile} holds for $\alpha=\frac34$ and $X,Y$ are not comonotonic. In the same example, note that we have $\E[X]=\E[Y]=1$, $\E[XY]=\frac32$, $\MAD(X)=\MAD(Y)=\frac23$, and $\MAD(X+Y)=\frac43$ so that \eqref{eq:add3} also holds. In particular, this example shows that \eqref{eq:add3} does not imply comonotonicity. Since $\V(X)=\V(Y)=\frac{2}{3}$ and $\Cov(X,Y)=\frac12$, $X$ and $Y$ are not perfectly positively correlated. Hence, \eqref{eq:add3} does not imply perfect positive correlation either. As another example, assuming that $\P\{X=0,Y=0\}=\P\{X=1,Y=1\}=\P\{X=2,Y=1\}=\frac13$, they show that \eqref{eq:expectile} does not hold for $\alpha=\frac34$ but $X,Y$ are comonotonic. Note that $\E[X]=1$, $\E[Y]=\frac23$, $\MAD(X)=\frac23$, $\MAD(Y)=\frac{4}{9}$, and $\MAD(X+Y)=\frac{11}{9}$ so that \eqref{eq:add3} does not hold either. In particular, comonotonicity does not imply \eqref{eq:add3}.
\end{remark}

As shown in Remarks~\ref{rem:corr},~\ref{rem:additivity}, two square-integrable random variables that satisfy \eqref{eq:add3} are positively correlated but not necessarily perfectly correlated. Indeed, the correlation coefficient of such random variables can be any value in $[0,1]$ as the next example shows.

\begin{example}
	Let $\mu_X,\mu_Y\in\R$ and $\rho\in[0,1]$. Suppose that there exist two square-integrable positive random variables $R_X,R_Y$ such that
	\[
	\frac{\E[R_XR_Y]}{\sqrt{\E[R_X^2]\E[R_Y^2]}}=\rho.
	\]
	Suppose also that there exists an event $A\in\F$ such that
	\begin{equation}\label{eq:RX}
	\E[R_X \mathbbm{1}_A]=\frac{\E[R_X]}{2},\quad \E[R_Y \mathbbm{1}_A]=\frac{\E[R_Y]}{2}.
	\end{equation}
	(For instance, any event $A\in\F$ that is independent of $R_X, R_Y$ and has $\P(A)=\frac12$ satisfies this condition.) Let $S:=\mathbbm{1}_A-\mathbbm{1}_{A^c}$ and define
	\[
	X:=\mu_X+SR_X,\quad Y:=\mu_Y+SR_Y.
	\]
	Thanks to \eqref{eq:RX}, we have $\E[SR_X]=\E[SR_Y]=0$ so that $\E[X]=\mu_X$, $\E[Y]=\mu_Y$. Hence,
	\[
	\Cov(X,Y)=\E[(X-\E[X])(Y-\E[Y])]=\E[S^2R_XR_Y]=\E[R_XR_Y]=\rho\sqrt{\E[R_X^2]\E[R_Y^2]},
	\]
	which implies that the correlation coefficient of $X,Y$ is $\rho$.
	\end{example}

The next theorem is a generalization of Lemma~\ref{MADAdditive} for the multivariate setting. It provides a list of equivalent statements that can be seen as the definition of additivity for MAD when multiple random variables are considered.

\begin{theorem} \label{th:AddGen}
  Let $Y_1, \ldots, Y_n\in L^1$. Then, the following are equivalent.
  \begin{enumerate}[label=(\roman*)]
  	\item $\MAD(\sum_{i=1}^n x_iY_i)=\sum_{i=1}^n x_i\MAD(Y_i)$ for every $\boldsymbol{x}\in\R^n_+$.
  	\item $\MAD(\sum_{i=1}^n x_iY_i)=\sum_{i=1}^n x_i\MAD(Y_i)$ for every $\boldsymbol{x}\in\Delta^{n-1}$.
  	\item $\MAD(\sum_{i\in I}Y_i)=\sum_{i\in I}\MAD(Y_i)$ for every nonempty subset $I\subseteq \{1,\ldots,n\}$.
  	\item $(Y_i-\E[Y_i])(Y_j-\E[Y_j])\geq 0$ $\P$-a.s. for every $i,j\in\{1,\ldots,n\}$ with $i\neq j$.
  	\end{enumerate}
\end{theorem}

\begin{proof}
 See Appendix \ref{app:Lemma}.
\end{proof}

\subsection{Subdifferential of MAD}

In order to introduce the marginal risk contributions of portfolios, we provide a theoretical treatment of the subdifferential of MAD as a function of the portfolio weights in this section. Suppose that there are $n$ assets in the market. We write $\boldsymbol{r}=(r_1,\ldots,r_n)^{\mathsf{T}}$ for the random vector of asset returns where we assume $r_i\in L^1$ for each $i\in\{1,\ldots,n\}$. We also write $\boldsymbol{\mu}=(\mu_1,\ldots,\mu_n)^{\mathsf{T}}$ for the corresponding mean vector of $\boldsymbol{r}$, i.e., $\E[r_i]=\mu_i$ for each $i\in\{1,\ldots,n\}$. Furthermore, we indicate by $\boldsymbol{x}= (x_{1}, \ldots, x_{n})^{\mathsf{T}}\in\R^n$ a vector of portfolio weights, which are the decision variables of the problems tackled in this work. For such $\boldsymbol{x}$, we have $\sum_{i=1}^n x_i =1$ and we denote by $R(\boldsymbol{x})$ the random portfolio return. Note that we have
\begin{equation*}
	R(\boldsymbol{x})=\boldsymbol{r}^{\mathsf{T}}\boldsymbol{x}=\sum_{i=1}^{n}r_{i}x_{i} \, .
\end{equation*}
Then, with some abuse of notation, the risk of the portfolio measured by MAD is given by
\begin{equation}\label{def:MAD}
	\MAD(\boldsymbol{x}):=\MAD(R(\boldsymbol{x}))=\E\sqb{\abs{(\boldsymbol{r}-\boldsymbol{\mu})^{\mathsf{T}}\boldsymbol{x}}}=
	\E\sqb{\abs{\sum_{i=1}^{n}(r_{i}-\mu_{i})x_{i}}}.
\end{equation}

%\subsubsection{Differentiability}\label{subsubsec:Differentiability}
%In this section, we define the marginal risk contribution of each asset to the MAD of the portfolio.
In \eqref{def:MAD}, we observe that the function $\boldsymbol{x}\mapsto \MAD(\boldsymbol{x})$ is generally not differentiable at some points in its domain due to the absolute value function. However, it is a finite, convex, hence also a continuous function on $\R^n$; thus, it has a nonempty subdifferential $\partial\MAD(\boldsymbol{x})$ at every given $\boldsymbol{x}\in\R^n$ as calculated by the next proposition.

\begin{proposition}\label{prop:subdifferentialMAD}
	Let $\bm{x}\in\R^n$. Then, the subdifferential of $\MAD$ at $\bm{x}$ is given by
	\[
	\partial \MAD(\bm{x})=\cb{\E\sqb{\of{ \sgn\of{( \bm{r}-\bm{\mu})^{\mathsf{T}}\bm{x}}+\eta \mathbbm{1}_{ \{(\bm{r}-\bm{\mu})^{\mathsf{T}}\bm{x}=0\}} } (\bm{r}-\bm{\mu})}\mid  \eta \in L^0,\ \eta\in [-1,1]\ \P\mbox{-a.s.}}.
	\]
	\end{proposition}

\begin{proof}
	By viewing $\MAD$ as an integral functional with respect to the probability measure $\P$, we may calculate its subdifferential using \citet[Theorem~23]{rockafellarintegral}, which works under the assumption that the probability space is complete. (See also \citet[Corollary, p.~179]{rockafellarwets1982} for a similar result with slightly different assumptions.) Note that we may write
	\[
	\MAD(\bm{x}) = \int_{\Omega}h(\omega,\bm{x})\P(d\omega), \quad \bm{x}\in\R^n,
	\]
	where $h\colon \O\times\R^n\to\R$ is defined by
	\[
	h(\omega,\bm{x}):=\abs{(\bm{r}(\omega)-\bm{\mu})^{\mathsf{T}}\bm{x}},\quad \omega\in\Omega,\bm{x}\in\R^n.
	\]
	For every $\bm{x}\in\R^n$, the function $\omega\mapsto h(\omega,\bm{x})$ on $\Omega$ is measurable since $\bm{r}$ is a random vector; for every $\omega\in\Omega$, the function $\bm{x}\mapsto h(\omega,\bm{x})$ is convex and continuous on $\R^n$. In particular, $h$ is measurable on $\Omega\times\R^n$. Let $\bm{\xi}\colon\O\to\R^n$ be a bounded random variable, we can find $M>0$ such that $\abs{\xi_i}\leq M$ for every $i\in\{1,\ldots,n\}$ $\P$-a.s. Then,
	\[
	\int_{\Omega}h(\omega,\bm{\xi}(\omega))\P(d\omega)=\int_{\Omega}\abs{(\bm{r}(\o)-\bm{\mu})^{\mathsf{T}}\bm{\xi}(\o)}\P(d\o)\leq M\sum_{i=1}^n \int_{\O}\abs{r_i(\o)-\mu_i}\P(d\o)<+\infty
	\]
	since $r_1,\ldots,r_n\in L^1$. Hence, all assumptions in \citet[Theorem~23]{rockafellarintegral} are satisfied and we obtain that
	\begin{equation}\label{eq:subdif}
	\partial\MAD(\bm{x})=\{\E[\bm{\xi}]\mid \bm{\xi}\in (L^1)^n,\ \P\{\o\in\O\mid \bm{\xi}(\o)\in \partial_{\bm{x}} h(\o,\bm{x})\}=1\},
	\end{equation}
	where, for each $\o\in\O$, $\partial_{\bm{x}}h(\o,\bm{x})$ denotes the subdifferential of the function $\bm{x}\mapsto h(\o,\bm{x})$ at $\bm{x}$. By standard rules of finite-dimensional subdifferential calculus (see, e.g., \citet[Theorem 23.9]{rockafellar1970convex}), we can calculate this set as
	\begin{equation}\label{eq:subdif2}
	\partial_{\bm{x}}h(\o,\bm{x})=(\bm{r}(\o)-\bm{\mu})\partial\abs{\cdot}((\bm{r}(\o)-\bm{\mu})^{\mathsf{T}}\bm{x}),
	\end{equation}
	where $\partial\abs{\cdot}$ denotes the (set-valued) subdifferential mapping of the absolute value function and is given by $\partial\abs{\cdot}(s)=\{\sgn(s)\}$ for $s\neq 0$, and by $\partial\abs{\cdot}(s)=[-1,+1]$ for $s=0$, which can be written compactly as
\begin{equation}\label{eq:subdif3}
\partial\abs{\cdot}(s)=\{\sgn(s)+t\mathbbm{1}_{\{0\}}(s)\mid t\in [-1,+1]\},\quad s\in\R.
\end{equation}
Combining \eqref{eq:subdif2}, \eqref{eq:subdif3}, we get
\[
\partial_{\bm{x}}h(\o,\bm{x})=\{(\sgn((\bm{r}(\o)-\bm{\mu})^{\mathsf{T}}\bm{x})+t\mathbbm{1}_{\{0\}}((\bm{r}(\o)-\bm{\mu})^{\mathsf{T}}\bm{x}))(\bm{r}(\o)-\bm{\mu})\mid t\in [-1,+1]\}.
\]
Combining this with \eqref{eq:subdif} yields that $\partial\MAD(\bm{x})$ is the set of all expectations of the form $\E[\bm{\xi}]$, where $\bm{\xi}\in (L^1)^n$ is such that
\begin{equation}\label{cond1}
\P\{\o\in\O\mid \exists t\in [-1,1]\colon \bm{\xi}(\o)=(\sgn((\bm{r}(\o)-\bm{\mu})^{\mathsf{T}}\bm{x})+t\mathbbm{1}_{\{(\bm{r}-\bm{\mu})^{\mathsf{T}}\bm{x}=0\}}(\o))(\bm{r}(\o)-\bm{\mu})\}=1.
\end{equation}
Here, $t$ depends on $\omega$ but its measurability is not guaranteed. To resolve this measurability issue, we show that condition \eqref{cond1} holds if and only if there exists $\eta\in L^0$ such that $-1\leq \eta\leq 1$ $\P$-a.s. and
\begin{equation}\label{cond2}
\P\{\o\in\O\mid  \bm{\xi}(\o)=(\sgn((\bm{r}(\o)-\bm{\mu})^{\mathsf{T}}\bm{x})+\eta(\omega)\mathbbm{1}_{\{(\bm{r}-\bm{\mu})^{\mathsf{T}}\bm{x}=0\}}(\o))(\bm{r}(\o)-\bm{\mu})\}=1.
\end{equation}
Indeed, the ``if" part of the claim is trivial. For the ``only if" part, suppose that \eqref{cond1} holds. Let $A\in\F$ be the event in the condition. Hence, the set
\[
T(\o):=\{t\in[-1,1]\mid \bm{\xi}(\o)=(\sgn((\bm{r}(\o)-\bm{\mu})^{\mathsf{T}}\bm{x})+t\mathbbm{1}_{\{(\bm{r}-\bm{\mu})^{\mathsf{T}}\bm{x}=0\}}(\o))(\bm{r}(\o)-\bm{\mu})\}
\]
is nonempty for every $\o\in A$. Moreover,
\[
(t,\o)\mapsto (\sgn((\bm{r}(\o)-\bm{\mu})^{\mathsf{T}}\bm{x})+t\mathbbm{1}_{\{(\bm{r}-\bm{\mu})^{\mathsf{T}}\bm{x}=0\}}(\o))(\bm{r}(\o)-\bm{\mu})
\]
is a Carath\'{e}odory function, i.e., it is continuous in $t\in[-1,1]$ and measurable in $\o\in\O$. 
For definiteness, let us set $t(\o)=0$ for $\o\in \O\setminus A$. The measurability of $\bm{\xi}$ together with \citet[Corollary~14.6, Example~14.15]{rockafellarwetsbook} imply that there exists a random variable $\eta$ such that $\eta(\o)\in T(\o)$ for almost every $\o\in\O$. Therefore, \eqref{cond2} holds for $\eta$, which yields the desired subdifferential formula for $\MAD$.
	\end{proof}

\section{Risk Parity approach with MAD}\label{sec:RPwithMAD}

The Risk Parity (RP) approach aims to choose the portfolio weights such that
%risk contributions of all assets to the global risk of the portfolio are equal.
the risk contributions of all assets are equal. In this section, we discuss the existence, uniqueness, and calculation of such portfolio weights.
\subsection{MAD-RP portfolios: existence and uniqueness}

The standard method for decomposing a risk measure (in the broad sense) into additive components is based on \emph{Euler's homogenous function theorem}. Since $\boldsymbol{x}\mapsto \MAD(\boldsymbol{x})$ is a homogeneous function of degree 1, if it were continuously differentiable, then we would have
\begin{equation}\label{eq:Euler}
\MAD(\boldsymbol{x})= \nabla \MAD(\boldsymbol{x})^{\mathsf{T}}\boldsymbol{x}=\sum_{i=1}^n x_i\frac{\partial \MAD(\boldsymbol{x})}{\partial x_i}.
\end{equation}
Similar to the case of volatility treated as a risk measure \citep{Maillard2010}, we could interpret the quantity
$x_{i} \frac{\partial \MAD(\boldsymbol{x})}{\partial x_{i}}$ as the risk contribution of the asset $i$ to the global MAD of the portfolio.

We need a suitable replacement of \eqref{eq:Euler} since MAD is not differentiable at some points of its domain in general. Nevertheless, by the generalized version of Euler's homogeneous function theorem for the subdifferentials of convex functions \cite[Theorem 3.1]{GenEuler}, we still have
\begin{equation}\label{eq:GenEuler}
\MAD(\boldsymbol{x})=\boldsymbol{s}^{\mathsf{T}}\boldsymbol{x},\quad \boldsymbol{s}\in \partial \MAD(\boldsymbol{x}).
\end{equation}
Given two vectors $\boldsymbol{a},\boldsymbol{b}\in\R^n$, we write $\boldsymbol{a}\cdot\boldsymbol{b}=(a_1b_1,\ldots,a_nb_n)^{\mathsf{T}}$ for their Hadamard product. Let us introduce the set
\begin{equation}\label{RCset}
\mathcal{RC}(\boldsymbol{x}):=\boldsymbol{x}\cdot\partial \MAD(\boldsymbol{x}):=\{\boldsymbol{x}\cdot \boldsymbol{s}(\boldsymbol{x})\mid \boldsymbol{s}(\boldsymbol{x})\in \partial\MAD(\boldsymbol{x})\}.
\end{equation}
Suppose that $\boldsymbol{x}\in\Delta^{n-1}$ is a portfolio vector. For each $\RC(\boldsymbol{x})=(\RC_i(\boldsymbol{x}),\ldots,\RC_n(\boldsymbol{x}))^{\mathsf{T}}\in\mathcal{RC}(\boldsymbol{x})$, \eqref{eq:GenEuler} yields that
\[
\MAD(\boldsymbol{x})=\sum_{i=1}^{n} \RC_{i}(\boldsymbol{x}).
\]
For this reason, we call $\RC(\boldsymbol{x})$ a \emph{feasible risk contribution vector} for portfolio $\boldsymbol{x}$. For instance, it is easy to verify that $\boldsymbol{x}\cdot \E[\sgn((\boldsymbol{r}-\boldsymbol{\mu})^{\mathsf{T}}\boldsymbol{x})(\boldsymbol{r}-\boldsymbol{\mu})]$ is a feasible risk contribution vector for $\boldsymbol{x}$; see Proposition~\ref{prop:subdifferentialMAD}.

%In particular, we may choose $\boldsymbol{s}:=\E[\sgn((\boldsymbol{r}-\boldsymbol{\mu})^{\mathsf{T}}\boldsymbol{x})(\boldsymbol{r}-\boldsymbol{\mu})]$ and it is easy to verify that $\boldsymbol{s}\in \partial \MAD(\boldsymbol{x})$ in this case; see \eqref{subdifferentialMAD}. For each $i\in\{1,\ldots,n\}$, let us define the risk contibution of asset $i$ as
%\begin{equation}
%\RC_{i}(\boldsymbol{x}):= x_i s_i=x_{i} \E[\sgn((\boldsymbol{r}-\boldsymbol{\mu})^{\mathsf{T}}\boldsymbol{x}) (r_{i}-\mu_{i})].
%\label{eq:TRCMAD}
%\end{equation}
%Then,
%
%Indeed, let us define the risk contribution of asset $i$ as
%%
%%
%Then, it is straightforward to show that
%%
%$
%\displaystyle\sum_{i=1}^{n} \RC_{i}(\boldsymbol{x}) = \MAD(\boldsymbol{x})
%$
%%

A MAD-RP portfolio is characterized by the
requirement of having the same risk contribution for each asset. More precisely, a long-only portfolio $\boldsymbol{x}\in\Delta^{n-1}$ is called a \emph{MAD-RP portfolio} if there exists $\RC(\boldsymbol{x})\in\mathcal{RC}(\boldsymbol{x})$ such that
\begin{equation}\label{eq:MADRPcond}
  \RC_{i}(\boldsymbol{x})= \RC_{j}(\boldsymbol{x})
\end{equation}
for each $i,j\in \{1,\ldots,n\}$ with $i\neq j$.
%
%\noindent where $TRC_{i}(\boldsymbol{x})=x_i\frac{\partial \phi(\boldsymbol{x})}{\partial x_i}$, and $\phi(x)$ is a generic risk measure and a homogeneous function. Since this model exploits Euler's Homogenous Function Theorem, we obtain $k\phi(\boldsymbol{x})=\sum\limits_{i=1}^{n}TRC_{i}(\boldsymbol{x})$, with $k$ being the degree of homogeneity.
Then, a long-only MAD-RP portfolio can be found by solving the system
\begin{equation}
	\left\{
	\begin{array}{ll}
		\RC_{i}(\boldsymbol{x})=\lambda, & \quad i\in\{1,\ldots,n\}, \\
		\boldsymbol{x}\in\Delta^{n-1},\\
		\RC(\boldsymbol{x})\in\mathcal{RC}(\boldsymbol{x}),\\
		\lambda \in\R.
	\end{array}%
	\right.  \label{eq:RPport1}
\end{equation}

For the results of this section, we work under the following nondegeneracy assumption.

\begin{assumption}\label{assumption}
	The only vector $\boldsymbol{x}\in\R^n_+$ such that $(\boldsymbol{r}-\boldsymbol{\mu})^{\mathsf{T}}\boldsymbol{x}=0$ $\P$-a.s. is $\boldsymbol{x}=0$.
	\end{assumption}

We discuss the mildness of Assumption~\ref{assumption} in the next two examples.

\begin{example}\label{ex1}
	Suppose that $\bm{r}-\bm{\mu}$ has distinct values $\bm{\nu}_1,\ldots,\bm{\nu}_m\in\R^n$ with respective probabilities $p_1,\ldots,p_m>0$, where $m\geq n$ and $\sum_{j=1}^m p_j=1$. In this case, the condition $(\boldsymbol{r}-\boldsymbol{\mu})^{\mathsf{T}}\boldsymbol{x}=0$ $\P$-a.s. is equivalent to the linear system $\bm{A}\bm{x}=\bm{0}_m$, where $\bm{A}\in\R^{m\times n}$ is the matrix whose respective rows are $\bm{\nu}_1^{\mathsf{T}}, \ldots, \bm{\nu}_m^{\mathsf{T}}$ and $\bm{0}_m$ is the $m$-dimensional zero vector. Then, the following are equivalent:
		\begin{enumerate}[label=(\roman*)]
			\item Assumption~\ref{assumption} holds.
			\item The only solution of the system $\bm{A}\bm{x}=\bm{0}_m, \bm{x}\in\R^n_+$ is $\bm{x}=0$.
			\item There exists $\bm{\lambda}\in\R^m$ such that $\bm{A}^{\mathsf{T}}\bm{\lambda}\in\R^n_{++}$.
			\item There exist $\lambda_1,\ldots,\lambda_m\in\R$ such that $\sum_{j=1}^m \lambda_j\bm{\nu}_j\in\R^n_{++}$.
			\end{enumerate}
		 Here, the equivalence between (ii) and (iii) is by Gordan's Alternative Theorem; see \citet{gordan} and \citet[Theorem 3.14]{guler}. In particular, Assumption~\ref{assumption} holds if the rank of $\bm{A}$ is $n$, i.e., there are $n$ linearly independent vectors among $\bm{\nu}_1,\ldots,\bm{\nu}_m$.
	\end{example}

\begin{example}\label{ex2}
	Suppose that $\bm{r}-\bm{\mu}$ has the multivariate centered Gaussian distribution with covariance matrix $\bm{\Sigma}\in\R^{n\times n}$. Then, for every $\bm{x}\in\R^n$, $(\boldsymbol{r}-\boldsymbol{\mu})^{\mathsf{T}}\boldsymbol{x}$ has the univariate centered Gaussian distribution with variance $\bm{x}^{\mathsf{T}}\bm{\Sigma}\bm{x}$. In this case, the condition $(\boldsymbol{r}-\boldsymbol{\mu})^{\mathsf{T}}\boldsymbol{x}=0$ $\P$-a.s. is equivalent to $\bm{x}^{\mathsf{T}}\bm{\Sigma}\bm{x}=0$. Hence, Assumption~\ref{assumption} holds if and only if $\bm{\Sigma}$ is strictly copositive. In particular, Assumption~\ref{assumption} holds if $\bm{\Sigma}$ is positive definite, i.e., $\bm{\Sigma}$ has full rank or, equivalently, $\bm{\Sigma}$ is nonsingular.
	\end{example}

%The following lemma is a technical preparation for Proposition~\ref{RPMAD-unique}.

%\begin{lemma}\label{RPMAD-lemma}
%	Suppose that Assumption~\ref{assumption} holds. 
%	\end{lemma}

%\begin{proof}
%	The supposition ensures that $\MAD(\boldsymbol{x})=\E[|(\boldsymbol{r}-\boldsymbol{\mu})^{\mathsf{T}}\boldsymbol{x}|]>0$ for every $\boldsymbol{x}\in\R^n_+\setminus\{0\}$, that is, $\MAD$ is a positive function as defined in \citet[Section 2]{Cesarone2019Anoptimization}. Moreover, $\MAD$ is also convex and positively homogeneous. Although $\MAD$ is not continuously differentiable at some points of its domain, the same arguments as in the proof of \citet[Theorem 2]{Cesarone2019Anoptimization} are applicable, and it follows that the given problem has a unique optimal solution.
%	\end{proof}

The next proposition establishes the existence and uniqueness of a MAD-RP portfolio.

\begin{proposition}\label{RPMAD-unique}
	 Suppose that Assumption~\ref{assumption} holds. 
	\begin{enumerate}[label=(\roman*)]
		\item Let $\lambda>0$ and consider the problem
		\begin{equation}
			\begin{array}{rl}
				\displaystyle\min_{\boldsymbol{x}\in\R^n_{++}} & \MAD(\boldsymbol{x})-\lambda\displaystyle\sum_{i=1}^{n}\ln x_{i}.\\
			\end{array}%
			\label{MADlambda}
		\end{equation}
		Then, there exists a unique optimal solution of this problem. 
	 \item There exists a unique MAD-RP portfolio, that is, the system \eqref{eq:RPport1} has a unique solution.
	 \end{enumerate}
	\end{proposition}

\begin{proof}
	(i) The supposition ensures that $\MAD(\boldsymbol{x})=\E[|(\boldsymbol{r}-\boldsymbol{\mu})^{\mathsf{T}}\boldsymbol{x}|]>0$ for every $\boldsymbol{x}\in\R^n_+\setminus\{0\}$, that is, $\MAD$ is a positive function as defined in \citet[Section 2]{Cesarone2019Anoptimization}. Moreover, $\MAD$ is also convex and positively homogeneous. Although $\MAD$ is not continuously differentiable at some points of its domain, the same arguments as in the proof of \citet[Theorem 2]{Cesarone2019Anoptimization} are applicable, and it follows that the given problem has a unique optimal solution.
	
	(ii) Let us fix some $\lambda>0$. By (i), there exists a unique optimal solution $\bar{\bm{x}}$ of the problem in \eqref{MADlambda}. By \citet[Proposition 2]{Cesarone2019Anoptimization}, $\boldsymbol{x}^\ast:=\frac{\bar{\boldsymbol{x}}}{\sum_{i=1}^n \bar{x}_i}$ is the unique optimal solution of the problem in \eqref{MADlambda} but with $\lambda$ replaced with $\lambda^\ast:=\frac{\lambda}{\sum_{i=1}^n \bar{x}_i}$. By the first order condition for optimality, we have
	\begin{equation}\label{eq:FOC}
	0\in \partial\MAD(\boldsymbol{x}^\ast)-\lambda^\ast \sqb{\frac{1}{x_i^\ast}}_{i=1}^n,
	\end{equation}
	that is,
	\[
	(\lambda^\ast,\ldots,\lambda^\ast)^{\mathsf{T}} \in \boldsymbol{x}^\ast\cdot \partial \MAD(\boldsymbol{x}^\ast).
	\]
	Hence, there exists $\boldsymbol{s}(\boldsymbol{x}^\ast)\in \partial \MAD(\boldsymbol{x}^\ast)$ such that $x^{\ast}_i s_i(\boldsymbol{x}^\ast)=\lambda^\ast$ for every $i\in\{1,\ldots,n\}$. In particular, $\RC(\boldsymbol{x}):=\boldsymbol{x}^\ast\cdot \boldsymbol{s}(\boldsymbol{x}^\ast)\in \mathcal{RC}(\boldsymbol{x}^\ast)$ and $\RC_i(\boldsymbol{x})=\lambda^\ast$ for every $i\in\{1,\ldots,n\}$. Therefore, $\boldsymbol{x}^\ast$ is a MAD-RP portfolio. Its uniqueness follows from the strict convexity of the logarithmic objective function in \eqref{MADlambda}.
	\end{proof}

\begin{remark}\label{RPMAD-rem}
	In \citet[Theorem~2]{Cesarone2019Anoptimization}, the existence and uniqueness of an RP portfolio are shown for a continuously differentiable risk measure. Their proof has two main components: 1) arguments based on coercivity and Weierstrass theorem to prove that the associated logarithmic problem (similar to \eqref{MADlambda}) has a unique optimal solution, where the positivity, convexity, and positive homogeneity of the risk measure are used but its differentiability is not used, 2) their Proposition~1, where it is shown that an optimal solution of the logarithmic problem gives rise to an RP portfolio and the proof uses first order conditions for the minimization of differentiable functions. In the present paper, 1) still works when the risk measure is replaced with MAD, as stated in the proof of Proposition~\ref{RPMAD-unique}(i). However, in the proof of Proposition~\ref{RPMAD-unique}(ii), we extend 2) by using the more general first order condition based on subdifferentials.
	\end{remark}

Assuming that the portfolio MAD is additive (see Theorem \ref{th:AddGen}(iv)), i.e., the asset returns satisfy the condition
\begin{equation}\label{eq:add4}
	(r_i-\mu_i)(r_j-\mu_j)\geq 0 \quad \P\mbox{-a.s.}
	\end{equation}
for every $i,j\in\{1,\ldots,n\}$ with $i\neq j$, we show in the next proposition that the weights of the assets in the MAD-RP portfolio are proportional to the reciprocals of the MADs of the individual asset returns.

We point out that taking a portfolio with weights proportional to the inverse of the risks of the individual asset returns is a na\"{\i}ve approach frequently used in practice. This approach is often called \emph{na\"{\i}ve risk parity} due to the strong implicit assumption on the dependence of the asset returns \citep[see, e.g.,][]{qian2011risk,Qian2017,clarke2013risk,haesen2017enhancing}.

\begin{proposition}[Closed-form solution for the long-only MAD-RP portfolio]\label{RPClosedForm}
Suppose that Assumption~\ref{assumption} holds. Furthermore, suppose that the additivity condition \eqref{eq:add4} holds. %Furthermore, for each asset $i\in\{1,\ldots,n\}$, suppose that its retun $r_i$ is not an almost surely deterministic random variable, that is, $\MAD(r_i)=\E[\abs{r_i-\mu_i}]>0$.
Then, the unique MAD-RP portfolio is given by
	\begin{equation}
	x_{i}=\displaystyle\frac{\E[\abs{r_{i}-\mu_{i}}]^{-1}}{\sum_{j=1}^{n}\E[ \abs{r_{j}-\mu_{j}}]^{-1}},\quad i\in\{1,\ldots,n\}.
	\label{eq:xRPMADAdditive}
	\end{equation}
	
	\begin{proof}
The first supposition ensures that there is a unique MAD-RP portfolio by Proposition \ref{RPMAD-unique}. Moreover, as checked in the proof of Proposition \ref{RPMAD-unique}, we have $\MAD(r_i)=\E[|r_i-\mu_i|]>0$ for each $i\in\{1,\ldots,n\}$. Hence, the expression in \eqref{eq:xRPMADAdditive} is well-defined.

Under \eqref{eq:add4}, MAD is additive for long-only portfolios by Theorem \ref{th:AddGen}. Hence, we have
\begin{equation}\label{eq:MADadd}
\MAD(\boldsymbol{x})= \sum_{i=1}^{n} x_{i} \MAD(r_{i})
\end{equation}
for every $\boldsymbol{x}\in\R^n_+$, i.e., $\MAD$ is linear on $\R^n_+$. It follows that $\MAD$ is differentiable at every $\boldsymbol{x}\in\R^n_{++}$ and the gradient is given by $\nabla \MAD(\boldsymbol{x})=(\MAD(r_1),\ldots,\MAD(r_n))^{\mathsf{T}}$. Hence, $\mathcal{RC}(\boldsymbol{x})$ is a singleton and the only feasible risk contribution vector is given by
\[
\RC(\boldsymbol{x})=(x_1\MAD(r_1),\ldots,x_n\MAD(r_n))^{\mathsf{T}}
\]
for each $\boldsymbol{x}\in\R^n_{++}$. %\footnote{Cagin: How do we conclude term-by-term equality? So far we have $\sum_{i=1}^n x_i \MAD(r_i)=\sum_{i=1}^n \RC_i(\boldsymbol{x})$ for ALL $x\in\R^n_+$. But $\RC_i$ depends on all components of $\boldsymbol{x}$. I couldn't see it.
%\noindent
%Francesco: since MAD is additive (under \eqref{eq:add4}), it is a linear function.
%Then, $\frac{\partial \MAD(\boldsymbol{x})}{\partial x_{i}} = \MAD(r_{i})$ and
%$\RC_i(\boldsymbol{x})=\RC_i(x_i)$.
%If this argumentation is not enough, maybe I have another way to follow...
%}
Imposing condition \eqref{eq:MADRPcond}, we can write
		\begin{equation}
		x_{i} \E[ \abs{r_{i}-\mu_{i}}] = x_{j} \E[ \abs{r_{j}-\mu_{j}}], \quad i,j\in\{1,\ldots,n\}.
		\end{equation}
Thus, for a fixed $i\in\{1,\ldots,n\}$, considering the full investment constraint $\sum_{j=1}^{n}x_{j}=1$, we obtain
\[
\sum_{j=1}^n x_j =\sum_{j=1}^n \frac{\E[\abs{r_i-\mu_i}]}{\E[\abs{r_j-\mu_j}]}x_i=\E[\abs{r_i-\mu_i}]x_i\sum_{j=1}^n \E[\abs{r_j-\mu_j}]^{-1}=1.
\]
From this, \eqref{eq:xRPMADAdditive} follows immediately. Hence, the portfolio given in \eqref{eq:xRPMADAdditive} is a MAD-RP portfolio, and this is the only such portfolio as established above.
\end{proof}
\end{proposition}

\subsection{MAD-RP portfolios: calculations}\label{sec:formulations}

In the following sections, we provide three methods for finding MAD-RP portfolios in practice. For these methods, as usual in portfolio optimization, we use a lookback approach where the possible realizations of the discrete random returns arise from historical data. This will ensure that our optimization problems are finite-dimensional. Let $T\in\N$ be the length of the time series of the prices of the assets. For each $i\in\{1,\ldots,n\}$ and $t\in\{0,\ldots,T\}$, we denote by $p_{it}$ the realized price of asset $i$ at time $t$, and for $t\geq 1$, we denote by $r_{it}=\frac{p_{it}-p_{i(t-1)}}{p_{i(t-1)}}$ the realized (linear) return of asset $i$ for the period ending at $t$. We assume that there are no ties of the outcomes and, therefore, that each historical scenario is equally likely with probability $\frac{1}{T}$ \citep[see, e.g.,][and references therein]{carleo2017approximating,cesarone2020computational}. In particular, for each $i\in\{1,\ldots,n\}$, the random return $r_i$ of asset $i$ is a discrete random variable with the uniform distribution over the set $\{r_{i1},\ldots,r_{iT}\}$ and we have $\mu_i=\frac{1}{T}\sum_{t=1}^T r_{it}$.

%Note that, as a special case of Theorem \ref{th:AddGen}, we may take $Y_i=r_i$, the random return of asset $i$, for each $i\in\{1,\ldots,n\}$. Then, the additivity condition (i) reads as $\MAD(\boldsymbol{x})=\sum_{i=1}^n x_i\MAD(r_i)$, and condition (iv) reads as
%\begin{equation}\label{eq:add4}
%	(r_{it}-\mu_i)(r_{jt}-\mu_j)\geq 0
%\end{equation}
%for every $t\in\{1,\ldots,T\}$ and $i,j\in\{1,\ldots,n\}$ with $i\neq j$.

Furthermore, we assume that the only vector $\boldsymbol{x}\in\R^n_+$ such that $(\boldsymbol{r}-\boldsymbol{\mu})^{\mathsf{T}}\boldsymbol{x}=0$ $\P$-a.s. is $\boldsymbol{x}=0$, which guarantees the existence and uniqueness of an MAD-RP portfolio by Proposition \ref{RPMAD-unique}.

\subsubsection{Logarithmic formulations}\label{logf}

The first formulation that we present consists of minimizing MAD with a logarithmic barrier term in the objective function \citep{bai2016least} as already given in \eqref{MADlambda}:
\begin{equation}
\left\{
\begin{array}{rl}
\displaystyle\min_{\boldsymbol{x}\in\R^n_{++}} & \MAD(\boldsymbol{x})-\lambda\displaystyle\sum_{i=1}^{n}\ln x_{i},\\
\end{array}%
\right.  \label{LambdaObj}
\end{equation}
where $\lambda>0$ is a constant. We name this strictly convex optimization problem \textbf{log\_obj}. Suppose that \textbf{log\_obj} has an optimal solution $\bar{\boldsymbol{x}}$ (due to the strict convexity of the objective function, it must be the unique optimal solution).
Then, following the arguments in the proof of Proposition \ref{RPMAD-unique}, the \emph{normalized} portfolio $\boldsymbol{x}^\ast=\frac{\bar{\boldsymbol{x}}}{\sum_{i=1}^{n}\bar{x}_{i}}$ is the MAD-RP portfolio and it is the unique optimal solution of \textbf{log\_obj} with $\lambda$ replaced with $\lambda^\ast:=\frac{\lambda}{\sum_{i=1}^n \bar{x}_i}$.

The second formulation has a logarithmic constraint and we call it \textbf{log\_constr} %; see
\citep[see][]{Spinu2013,cesarone2018minimum,bellini2021risk}. %\cite{cesarone2018minimum}, \cite{bellini2021risk}.
%It allows us to determine the features of the Risk Parity portfolios under two particular settings.
It consists of solving the following convex problem:
\begin{equation}
\left\{
\begin{array}{rl}
\displaystyle\min_{\boldsymbol{x}\in\R^n_{++}} & \MAD(\boldsymbol{x})\\
&   \displaystyle\sum_{i=1}^{n}\ln x_{i}\geq c,
\end{array}
\right.  \label{LogConstr}
\end{equation}
where $c\in\R$ is a constant.
%
%Clearly, to obtain \emph{normalized} portfolio weights, we can compute $\frac{\boldsymbol{x}}{\sum_{i=1}^{n}x_{i}}$ whenever $\boldsymbol{x}$ is feasible for \eqref{LogConstr}.
%
Moreover, if $\bar{\boldsymbol{x}}$ is an optimal solution of \textbf{log\_constr}, then $\frac{\bar{\boldsymbol{x}}}{\sum_{i=1}^{n}\bar{x}_{i}}$ is the MAD-RP as checked in the next remark.

\begin{remark}%[MP???\footnote{Mustafa! This is a part that Jacopo included in his PhD thesis, but perhaps it could be omitted. What do you think? }]
  %\noindent
  Since the function $\boldsymbol{x}\mapsto\MAD(\boldsymbol{x})$ is not differentiable everywhere, we can apply the general Karush-Kuhn-Tucker conditions with subdifferentials as formulated in \citet[Theorem 28.3]{rockafellar1970convex}. Note that the Lagrangian of \textbf{log\_constr} is given by
\begin{equation}
L(\boldsymbol{x}, \lambda)=\MAD(\boldsymbol{x})+\lambda\biggl(c-\sum_{i=1}^{n}\ln(x_{i}) \biggr)
\end{equation}
for each $\boldsymbol{x}\in\R^n_{++}$, $\lambda\geq 0$. Then, for a fixed $\lambda\geq 0$, the subdifferential of $\boldsymbol{x}\mapsto L(\boldsymbol{x},\lambda)$ at $\boldsymbol{x}\in\R^n_{++}$ is given by
\[
\partial_{\boldsymbol{x}}L(\boldsymbol{x},\lambda)=\partial\MAD(\boldsymbol{x})-\lambda \sqb{\frac{1}{x_i}}_{i=1}^n.
\]
In particular, the problem $\inf_{\boldsymbol{x}\in\R^n_{++}} L(\boldsymbol{x},\lambda)$ of calculating the dual objective function at a given $\lambda> 0$ is equivalent to \textbf{log\_obj}. Let $\bar{\boldsymbol{x}}$ be an optimal solution of \textbf{log\_constr} with $\bar{\lambda}\geq 0$ denoting the Lagrange multiplier of the logarithmic constraint at optimality. Similar to \eqref{eq:FOC}, we have $0\in \partial_{\boldsymbol{x}}L(\bar{\boldsymbol{x}},\bar{\lambda})$. Arguing as in the proof of Proposition \ref{RPMAD-unique}, we obtain the existence of $\boldsymbol{s}(\bar{\boldsymbol{x}})\in\partial\MAD(\bar{\boldsymbol{x}})$ such that $\bar{x}_is_i(\bar{\boldsymbol{x}})=\bar{\lambda}$ for every $i\in\{1,\ldots,n\}$. In particular, $\bar{\boldsymbol{x}}\cdot \boldsymbol{s}(\bar{\boldsymbol{x}})\in\RC(\bar{\boldsymbol{x}})$ and $\MAD(\bar{\boldsymbol{x}})=n\bar{\lambda}$. Since $\MAD(\bar{\boldsymbol{x}})>0$ by the assumption stated at the beginning of this section, we have $\bar{\lambda}>0$. On the other hand, by complementary slackness, we have $\bar{\lambda}(c-\sum_{i=1}^n\ln(\bar{x}_i))=0$. Hence, $c=\sum_{i=1}^n \ln(\bar{x}_i)$. It follows that the normalized version $\boldsymbol{x}^\ast:=\frac{\bar{\boldsymbol{x}}}{\sum_{i=1}^n \bar{x}_i}$ is the MAD-RP portfolio and it is an optimal solution of \textbf{log\_constr} with $c$ replaced with $c^\ast:=\sum_{i=1}^n \ln(x^\ast_i)=\sum_{i=1}^n \ln(\bar{x}_i)-n\ln(\sum_{i=1}^n \bar{x}_i)=c-n\ln(\sum_{i=1}^n \bar{x}_i)$.
%the subdifferential at $\boldsymbol{y}$ of the objective function w.r.t $y_{i}$ is equal to Equation \eqref{subdifferentialMAD}, while the subdifferential at $\boldsymbol{y}$  of the constraint w.r.t $y_{i}$ is equal to $\frac{{1}}{y_{i}}\mbox{ with } i=1,\ldots,n$\footnote{Note that since the natural logarithm is always differentiable, in this case, the subgradient and the subdifferential coincide with the differential and with the gradient, respectively.}.

\end{remark}

Formulation \eqref{LogConstr} allows us to obtain some features of the Risk Parity portfolio $\boldsymbol{x}^{MAD-RP}$ related to a global minimizer $\boldsymbol{x}^{MinMAD}$ of MAD (i.e., the MinMAD portfolio) and the EW portfolio $\boldsymbol{x}^{EW}$, as shown in Remarks \ref{rem:MinMADvsMAD-ERC} and \ref{rem:EWvsMAD-ERC}, respectively.

\begin{remark}[MinMAD vs. MAD-RP] \label{rem:MinMADvsMAD-ERC}

From the formulation \textbf{log\_constr} in \eqref{LogConstr}, it is clear that
\[
\MAD(\boldsymbol{x}^{MinMAD})\leq \MAD(\boldsymbol{x}^{MAD-RP}).
\]
\end{remark}

\begin{remark}[EW vs. MAD-RP] \label{rem:EWvsMAD-ERC}
%When $c=-n\ln(n)$ the Risk Parity with MAD portfolio is equal to the equally weighted one.
Let us consider a slightly modified version of Problem \eqref{LogConstr},
where we add the budget constraint $\displaystyle\sum_{i=1}^{n} x_{i}=1$.
Due to Jensen inequality, we have
\begin{equation}\label{JensenRP}
\displaystyle\frac{1}{n}\sum_{i=1}^{n}\ln(x_{i})\leq \ln\biggl(\frac{1}{n}\sum_{i=1}^{n}x_{i}\biggr) = - \ln(n)
\end{equation}
for every feasible solution $\boldsymbol{x}$ of the new problem. Now, fixing $c=-n\ln(n)$ in Problem \eqref{LogConstr}, the unique feasible solution is given by $x_{i}=\frac{1}{n}$, i.e., $\boldsymbol{x}=\boldsymbol{x}^{EW}$. As a consequence, we have
$$\MAD(\boldsymbol{x}^{MAD-RP})\leq \MAD(\boldsymbol{x}^{EW}).$$
\end{remark}

\subsubsection{System-of-equation formulations}\label{soef}

An alternative method, named \textbf{soe\_1}, consists of solving the system \eqref{eq:RPport1} directly. Following the expression of the set $\mathcal{RC}(\boldsymbol{x})$ in \eqref{RCset} and that of the subdifferential $\partial \MAD(\boldsymbol{x})$ in Proposition~\ref{prop:subdifferentialMAD}, we can rewrite this system as
\begin{equation}
\negthinspace\negthinspace\negthinspace\negthinspace\negthinspace\left\{\negthinspace
\begin{array}{ll}
\frac{x_{i}}{T}\sum_{t=1}^T s_t(r_{it}-\mu_{i}) = \displaystyle \lambda, & i\in \{1,\ldots,n\},\\
\sum\limits_{i=1}^{n}x_{i}=1, &  \\
x_{i}\geq 0, & i\in\{1,\ldots,n\}, \\
\sum_{i=1}^n (r_{it}-\mu_{i})x_i\neq 0 \Rightarrow s_t=\sgn(\sum_{i=1}^n (r_{it}-\mu_{i})x_i), & t\in \{1,\ldots,T\},\\
-1\leq s_t\leq +1,& t\in\{1,\ldots,T\}, \\
\lambda\in\R, \; x_i\in\R,\; s_t\in \R,\; & i\in\{1,\ldots,n\}, t\in\{1,\ldots,T\}.\\
\end{array}%
\right.  \label{eq:SoE1}
\end{equation}
Note that $s_t$ is an auxiliary decision variable in this formulation and it stands for a subgradient of the absolute value function at the point $\sum_{i=1}^n (r_{it}-\mu_{i})x_i$. The logical implication constraint in the above formulation can be converted into linear inequalities by introducing additional variables. We rewrite $s_t$ as a convex combination of $+1$ and $-1$, i.e., $s_t=\alpha_t-(1-\alpha_t)=2\alpha_t-1$ for some variable $\alpha_t$ with values in $[0,1]$, where $\alpha_t$ is enforced to be $1$ when $\sum_{i=1}^n (r_{it}-\mu_{i})x_i>0$, and it is enforced to be $0$ when $\sum_{i=1}^n (r_{it}-\mu_{i})x_i<0$. This is achieved by the following formulation with the help of two additional binary decision variables:
\begin{equation}
\left\{
\begin{array}{ll}
\frac{x_{i}}{T} \sum_{t=1}^T (2\alpha_t-1)(r_{it}-\mu_{i}) = \displaystyle \lambda, & \quad i\in \{1,\ldots,n\},\\
\sum\limits_{i=1}^{n}x_{i}=1, & \\
x_{i}\geq 0, &\quad i\in\{1,\ldots,n\}, \\
-Mu_t\leq \sum_{i=1}^n (r_{it}-\mu_{i})x_i \leq Mv_t,&\quad t\in\{1,\ldots,T\},\\
1-M(1-v_t)\leq \alpha_t\leq M(1-u_t),&\quad t\in\{1,\ldots,T\},\\
0\leq \alpha_t\leq 1,&\quad t\in\{1,\ldots,T\}, \\
u_t,v_t\in\{0,1\},&\quad t\in\{1,\ldots,T\},\\
\lambda\in\R; \; x_i\in\R;\; \alpha_t,u_t,v_t\in\R,&\quad i\in\{1,\ldots,n\},\ t\in\{1,\ldots,T\},
\end{array}%
\right.  \label{eq:SoEnew1}
\end{equation}
where $M>0$ is a suitably large constant.
%
%\begin{equation}
%\left\{
%\begin{array}{ll}
%x_{i} \E\biggl[ \sgn\biggl[ \sum_{j=1}^{n}(r_{j}-\mu_{j})x_{j}\biggr](r_{i}-\mu_{i})\biggr] = \displaystyle \frac{z}{n} & \quad i=1,\ldots,n\\
%\sum\limits_{i=1}^{n}x_{i}=1 &  \\
%x_{i}\geq 0 & \quad i=1,\ldots ,n
%\end{array}%
%\right.  \label{eq:SoE1}
%\end{equation}
%
%\noindent To solve Problem \eqref{eq:SoE1}, we consider the following mixed-integer nonlinear problem (MINLP) using additional binary variables $y,q,t$, where the sign function can be approximated by $y+t$. Thus, we have:
%\begin{equation}
%\left\{
%\begin{array}{ll}
%x_{i} \E\biggl[ (y+t)(r_{i}-\mu_{i})\biggr] = \displaystyle \frac{z}{n} & \quad i=1,\ldots,n\\	
%\sum_{j=1}^{n}(r_{j}-\mu_{j})x_{j} \leq M y\\	
%-\sum_{j=1}^{n}(r_{j}-\mu_{j})x_{j} \leq M q\\
%\sum\limits_{i=1}^{n}x_{i}=1 &  \\
%y + q \leq 1\\
%t = -q\\
%x_{i}\geq 0 & \quad i=1,\ldots ,n
%\end{array}%
%\right.  \label{eq:SoEnew1}
%\end{equation}
% This reformulation can be processed numerically using the MINLP solver in TOMLAB.}
This reformulation is processed numerically using the global optimization software BARON \citep{tawarmalani2005polyhedral}.
\begin{remark}\label{rem:simplified}
	As a simplification of \eqref{eq:SoEnew1}, one may consider the following system:
	\begin{equation}
	\left\{
	\begin{array}{ll}
	\frac{x_{i}}{T} \sum_{t=1}^T (v_t-u_t)(r_{it}-\mu_{i}) = \displaystyle \lambda, & \quad i\in \{1,\ldots,n\},\\	
	\sum\limits_{i=1}^{n}x_{i}=1, &  \\
	x_{i}\geq 0, & \quad i\in\{1,\ldots ,n\},\\
	-Mu_t\leq \sum_{i=1}^n (r_{it}-\mu_{i})x_i \leq M v_t,&\quad t\in\{1,\ldots,T\},\\	
	v_t + u_t \leq 1,&\quad t\in\{1,\ldots,T\},\\
	u_t,v_t\in \{0,1\},&\quad t\in\{1,\ldots,T\},\\
	\lambda\in\R;\; x_i\in\R;\; u_t,v_t\in \R,&\quad i\in\{1,\ldots,n\},\ t\in\{1,\ldots,T\}.
	\end{array}%
	\right.  \label{eq:SoEnew2}
	\end{equation}
	In this system, $v_t-u_t$ stands for a subgradient of the absolute value function at a point where we have $\sum_{i=1}^n (r_{it}-\mu_{i})x_i $; however, it is restricted to the set $\{-1,0,+1\}$ when $\sum_{i=1}^n (r_{it}-\mu_{i})x_i =0$. Hence, for some corner cases, the system \eqref{eq:SoEnew2} may be infeasible even though the MAD-RP portfolio exists. On the other hand, if $\boldsymbol{x}^\ast$ is part of a solution of \eqref{eq:SoEnew2}, then it is the MAD-RP portfolio. For this reason, we will use \eqref{eq:SoEnew2} in the empirical analysis of Section \ref{sec:EmpirCompRes} when implementing \textbf{soe\_1}.
%A similar simplification is used for \textbf{soe\_2}.
	\end{remark}

%Finally, to remove the sign restrictions for $\boldsymbol{x}$ in \eqref{eq:SoE1}, we can also consider the following alternative formulation, which we call \textbf{soe\_2}:
%
Finally, we also consider the following formulation, which we call \textbf{soe\_2} and consists of a further simplified version of the system \eqref{eq:RPport1}
\begin{equation}
		\left\{
		\begin{array}{ll}
			q_{i}^{2} \E\biggl[ \sgn\biggl[ \sum_{j=1}^{n}(r_{j}-\mu_{j})q_{j}^{2}\biggr](r_{i}-\mu_{i})\biggr] = \displaystyle \frac{z}{n}, & \quad i\in\{1,\ldots,n\},\\
			\sum\limits_{i=1}^{n}q_{i}^{2}=1, &
		\end{array}%
		\right.  \label{eq:SoE2old}
	\end{equation}
where we remove the sign restrictions for $\boldsymbol{x}$.
%, substituting $x_{i}=q_{i}^{2}$.
%, and $z$ is a decision variable.
%
%\begin{equation}
%\left\{
%\begin{array}{ll}
%\frac{q^2_{i}}{T} \sum_{t=1}^T(2\alpha_t-1)(r_{it}-\mu_{i}) = \displaystyle \lambda, & \quad i\in \{1,\ldots,n\},\\
%\sum\limits_{i=1}^{n}q^2_{i}=1, & \\
%-Mu_t\leq \sum_{i=1}^n (r_{it}-\mu_i)q_i^2\leq Mv_t,&\quad t\in\{1,\ldots,T\},\\
%1-M(1-v_t)\leq \alpha_t\leq M(1-u_t),&\quad t\in\{1,\ldots,T\},\\
%0\leq \alpha_t\leq 1,&\quad t\in\{1,\ldots,T\}, \\
%u_t,v_t\in\{0,1\},&\quad t\in\{1,\ldots,T\},\\
%\lambda\in\R;\; q_i\in\R;\; \alpha_t, u_t, v_t\in\R,&\quad i\in\{1,\ldots,n\},\ t\in\{1,\ldots,T\}.
%\end{array}%
%\right.  \label{eq:SoE2}
%\end{equation}
%\begin{equation}
%\left\{
%\begin{array}{ll}
%q_{i}^{2} \E\biggl[ \sgn\biggl[ \sum_{j=1}^{n}(r_{j}-\mu_{j})q_{j}^{2}\biggr](r_{i}-\mu_{i})\biggr] = \displaystyle \frac{z}{n} & \quad i=1,\ldots,n\\
%\sum\limits_{i=1}^{n}q_{i}^{2}=1 &
%%
%\end{array}%
%\right.  \label{eq:SoE2}
%\end{equation}

\subsubsection{Least-squares formulations}\label{lsf}

The last method, called \textbf{ls\_rel}, exploits the fact that
we have $\frac{\RC_{i}(\boldsymbol{x})}{\MAD(\boldsymbol{x})}=\frac{1}{n}$ for each $i\in\{1,\ldots,n\}$ whenever $\boldsymbol{x}$ is the MAD-RP portfolio \citep[see][]{cesarone2018minimum}.
More precisely, it consists of the following least-squares formulation:
\begin{equation}
\left\{
\begin{array}{rll}
\displaystyle\min & \displaystyle\sum_{i=1}^{n}\biggl(\frac{\RC_{i}}{\MAD(\boldsymbol{x})} -\frac{1}{n}\biggr)^{2} & \\
& \boldsymbol{\RC}\in \mathcal{RC}(\boldsymbol{x}),&\\
&   \sum\limits_{i=1}^{n}x_{i}=1, &  \\
& x_{i}\geq 0, & \quad i\in \{1,\ldots ,n\},\\
&\boldsymbol{x}\in\R^n,\; \boldsymbol{\RC}\in \R^n.
\end{array}%
\right.  \label{RPLQRel}
\end{equation}

A variant of this least-squares method is given by
\begin{equation}
\left\{
\begin{array}{rll}
\displaystyle\min& \displaystyle\sum_{i=1}^{n}\sum_{j=1}^n\bigl(\RC_{i}-\RC_{j}\bigr)^{2} & \\
& \boldsymbol{\RC}\in \mathcal{RC}(\boldsymbol{x}),&\\
&   \sum\limits_{i=1}^{n}x_{i}=1, &  \\
& x_{i}\geq 0, & \quad i\in\{ 1,\ldots ,n\},\\
& \boldsymbol{x}\in\R^n,\; \boldsymbol{\RC}\in \R^n,
\end{array}%
\right.  \label{RPLQAbs}
\end{equation}
named \textbf{ls\_abs} \citep[see, e.g.,][]{Maillard2010}.
The functional constraint $\boldsymbol{\RC}\in \mathcal{RC}(\boldsymbol{x})$ can be reformulated in terms of linear inequalities with additional binary variables as in Section \ref{soef}. If an optimal solution for \textbf{ls\_rel} or \textbf{ls\_abs} is found with zero optimal value, then this solution is the MAD-RP portfolio. Nevertheless, numerical methods can only guarantee local optimality due to the nonconvex nature of these problems. Hence, using the least-squares formulations should be considered as a heuristic approach.

\section{Empirical analysis} \label{sec:EmpirCompRes}

In this section, we provide an extensive empirical analysis based on three investment universes which are described in Section \ref{sec:datasets}.
More precisely,
in Section \ref{sec:AccuEffi}, we test and compare all the MAD-RP formulations in terms of accuracy and efficiency.
Section \ref{sec:insample} shows some graphical examples of the portfolio selection strategies analyzed in terms of portfolio weights and relative risk distribution among assets, using both volatility and MAD as risk measures.
In Section \ref{sec:out-of-sample}, we report and discuss their out-of-sample performance results based on a Rolling Time Window scheme of evaluation.

\subsection{Datasets} \label{sec:datasets}

We give here a brief description of the three real-world datasets used in the empirical analysis.
These datasets were obtained from Refinitiv and are publicly available on \\
\href{https://www.francescocesarone.com}{https://www.francescocesarone.com}.
\begin{itemize}
%	\item \emph{Dow Jones 2005} (DJIA2005): it consists of the daily prices in US dollars of 21 of the stocks that belonged to the Dow Jones index in that year, as if they were still in the index now. In this way, we do not only consider the best stocks, i.e., the ones who survived, but also the bad ones.
	
	\item \emph{ETF Emerging countries} (ETF-EC): it consists of the total return index\footnote{It consists in the price of the asset including dividends, assuming these dividends are reinvested in the company. \\
See, e.g., https:\/\/www.msci.com\/eqb\/methodology\/meth\_docs\/MSCI\_May12\_IndexCalcMethodology.pdf}, expressed in US dollars, of the ETFs of 24 emerging countries.
	
	\item \emph{Euro bonds} (EuroBonds): it consists of the total return index in euros of the government bonds of 11 countries belonging to the Eurozone, with maturities ranging from 1 year to 30 years.
	
	\item \emph{Commodities and Italian bonds mix} (CIB-mix): it consists of the total return index, expressed in euros, of a mixture of Italian government bonds, whose maturity ranges from 2 to 30 years, and commodities (agriculture, gold, energy and industrial metals).
	
%	\item \emph{World bonds mix} (WorldBonds): it consists of the total return index, expressed in euros, of a mixture of government bonds of the same countries as in the EuroBonds dataset, and of government bonds of six other countries (Australia, Canada, China, Japan, UK, USA).
\end{itemize}

In Table \ref{tab:Datasets}, we provide additional information about these datasets.
\begin{table}[htbp]
	\centering
	\vspace{5pt}
	{\small{\begin{tabular}{r|ccc}
				Dataset & \# Assets & Days & From-To \\\hline
%				DJIA2005 & 21 & 1564 & 1/2013-12/2018\\
				ETF-EC & 24 & 1042 & 1/2015-12/2018\\
				EuroBonds & 62 & 1564 & 1/2013-12/2018\\
				CIB-mix & 11 & 1564 & 1/2013-12/2018\\
%				WorldBonds & 104 & 1564 & 1/2013-12/2018\\
				\hline
			\end{tabular}}}
			\caption{List of the datasets analyzed.}
			\label{tab:Datasets}
		\end{table}

		\subsection{Computational results}
		
		\subsubsection{Accuracy and efficiency of the MAD-RP formulations} \label{sec:AccuEffi}
		
		In this section, we test and compare all the MAD-RP formulations, described in Section
		\ref{sec:formulations}, in terms of accuracy and efficiency using the
		following quantities:
		\begin{itemize}
			\item the square root of the value of the objective function of Problem \eqref{RPLQRel}, named $\sqrt{F(\boldsymbol{x})}$;
			\item the portfolio $\MAD$ obtained from each method, named $\MAD(\boldsymbol{x})$;
			\item the empirical mean of the absolute deviations of relative risk contributions from $\frac{1}{n}$,
			$$\mbox{MeanAbsDev} = \frac{1}{n}\sum_{i=1}^{n}\abs{ \frac{\RC_{i}(\boldsymbol{x})}{\MAD(\boldsymbol{x})}-\frac{1}{n}} \, ;$$
			\item the maximum of the absolute deviations of relative risk contributions from $\frac{1}{n}$,
			$$\mbox{MaxAbsDev} = \max_{1\leq i\leq n}\abs{ \frac{\RC_{i}(\boldsymbol{x})}{\MAD(\boldsymbol{x})}-\frac{1}{n}} \, ;$$
			\item the inverse of the number of assets selected, namely $\frac{1}{n}$;
			\item the computational times (in seconds) required to solve each MAD-RP formulation considered.

		\end{itemize}
		%
		%\noindent
		All the experiments have been implemented on a workstation with Intel Core
		CPU (i7-6700, 3.4 GHz, 16 Gb RAM) under MS Windows 10, using MATLAB 9.1.
More precisely, we solve models \textbf{log\_obj} \eqref{LambdaObj}, \textbf{log\_constr} \eqref{LogConstr}, \textbf{ls\_rel} \eqref{RPLQRel}, \textbf{ls\_abs} \eqref{RPLQAbs} by means of the built-in function \texttt{fmincon}; model \textbf{soe\_1} \eqref{eq:SoEnew2}
by using the global optimization software BARON (version 15.9.22), which is also called from MATLAB
by means of the MATLAB/BARON toolbox \citep{tawarmalani2005polyhedral};
%the \texttt{baron} function of the Baron/MATLAB interface, and,
finally, we solve \textbf{soe\_2} \eqref{eq:SoE2old} by means of the built-in function \texttt{fsolve}.
%
% Table generated by Excel2LaTeX from sheet 'ETF-EC'
\begin{table}[htbp]
  \centering
  \caption{Experimental results for the MAD-RP methods using ETF-EC, where $1/n = 0.0417$.}
  \medskip
  \scalebox{0.9}{
    \begin{tabular}{|c|c|c|c|c|c|}
    \toprule
    \textbf{Formulation} & $\boldsymbol{\sqrt{F(x)}}$ & \textbf{MAD(x)} & \textbf{MeanAbsDev} & \textbf{MaxAbsDev} & \textbf{Time (secs.)} \\
    \midrule
    \textbf{ls\_rel} & 3.33e-03 & 6.25e-03 & 5.19e-04 & 1.74e-03 & 4.76 \\
    \midrule
    \textbf{ls\_abs} & 2.52e-03 & 6.24e-03 & 4.42e-04 & 9.37e-04 & 4.74 \\
    \midrule
    \textbf{log\_constr} & 1.49e-03 & 6.23e-03 & 1.96e-04 & 8.36e-04 & 5.92 \\
    \midrule
    \textbf{log\_obj} & 1.49e-03 & 6.23e-03 & 1.95e-04 & 8.36e-04 & 2.51 \\
    \midrule
    \textbf{soe\_1} & -     & -     & -     & -     & $>$ 1 day \\
    \midrule
    \textbf{soe\_2} & 7.87e-03 & 6.29e-03 & 8.16e-04 & 7.28e-03 & 37.08 \\
    \bottomrule
    \end{tabular}%
    }
  \label{tab:RPMADETF-EC}%
\end{table}%

Since in Section \ref{sec:out-of-sample} we perform an extensive out-of-sample analysis using in-sample windows of 2 years, for the following accuracy and efficiency analysis of the MAD-RP formulations we set the number of observations $T=500$ days (i.e., the first two years of each dataset).
For Problems \eqref{LambdaObj} and \eqref{LogConstr}, we also perform a sensitivity analysis by varying the constants $\lambda >0$ and $c \in \R$, respectively.
Such a sensitivity analysis leads us to set $\lambda = 0.001$ and $c = -40$, which guarantee a reasonable level of accuracy for all experiments.
%
%
% Table generated by Excel2LaTeX from sheet 'EuroBonds'
\begin{table}[htbp]
  \centering
  \caption{Experimental results for the MAD-RP methods using EuroBonds, where $1/n = 0.0161$.}
  \medskip
  \scalebox{0.9}{
    \begin{tabular}{|c|c|c|c|c|c|}
    \toprule
    \textbf{Formulation} & $\boldsymbol{\sqrt{F(x)}}$ & \textbf{MAD(x)} & \textbf{MeanAbsDev} & \textbf{MaxAbsDev} & \textbf{Time (secs.)} \\
    \midrule
    \textbf{ls\_rel} & 2.54e-04 & 6.48e-04 & 2.36e-05 & 1.11e-04 & 62.18 \\
    \midrule
    \textbf{ls\_abs} & 5.19e-04 & 6.50e-04 & 4.87e-05 & 2.36e-04 & 65.22 \\
    \midrule
    \textbf{log\_constr} & 8.94e-05 & 6.47e-04 & 7.43e-06 & 6.15e-05 & 28.32 \\
    \midrule
    \textbf{log\_obj} & 9.02e-05 & 6.47e-04 & 7.75e-06 & 6.01e-05 & 5.43 \\
    \midrule
    \textbf{soe\_1} & -     & -     & -     & -     & $>$ 1 day \\
    \midrule
    \textbf{soe\_2} & 5.88e-02 & 1.19e-03 & 6.57e-03 & 1.45e-02 & 131.86 \\
    \bottomrule
    \end{tabular}%
    }
  \label{tab:RPMADEuroBonds}%
\end{table}%			

\noindent
Tables %\ref{tab:RPMADDJIA2005},
\ref{tab:RPMADETF-EC},
\ref{tab:RPMADEuroBonds}, and
\ref{tab:RPMADCIB-mix}
%		\ref{tab:RPMADWorldBonds}
report the computational results obtained for the three datasets analyzed.
Both \textbf{soe\_1} and \textbf{soe\_2} seem to be the least accurate and efficient formulations.
Note that, with this experimental setup, the most computationally demanding formulation \textbf{soe\_1} fails to find a solution in a reasonable time (1 day).

%The most likely reasons for the low precision and efficiency, in comparison with the other methods, are the large number of variables and constraints, and the type of problem. Indeed, while all the other methods consist in continuous problems, in this case, we have binary variables, as well, therefore obtaining a mixed-integer problem.
%
Methods based on least-squares formulations, i.e., \textbf{ls\_rel} and \textbf{ls\_abs}, achieve intermediate performances for all the datasets.
%; except for the EuroBonds dataset, the \textbf{ls\_abs} method proves to be the most precise and efficient between the two, as shown by, respectively, the MeanAbsDev and MaxAbsDev, and the Time values in Tables from \ref{} to \ref{}.
%
% Table generated by Excel2LaTeX from sheet 'CIB-mix'
\begin{table}[htbp]
  \centering
  \caption{Experimental results for the MAD-RP methods using CIB-mix, where $1/n = 0.0909$.}
  \medskip
  \scalebox{0.9}{
    \begin{tabular}{|c|c|c|c|c|c|}
    \toprule
    \textbf{Formulation} & $\boldsymbol{\sqrt{F(x)}}$ & \textbf{MAD(x)} & \textbf{MeanAbsDev} & \textbf{MaxAbsDev} & \textbf{Time (secs.)} \\
    \midrule
    \textbf{ls\_rel} & 2.48e-03 & 1.40e-03 & 6.16e-04 & 1.54e-03 & 2.66 \\
    \midrule
    \textbf{ls\_abs} & 2.32e-03 & 1.40e-03 & 5.63e-04 & 1.54e-03 & 2.09 \\
    \midrule
    \textbf{log\_constr} & 7.75e-04 & 1.40e-03 & 2.11e-04 & 3.48e-04 & 2.52 \\
    \midrule
    \textbf{log\_obj} & 7.74e-04 & 1.40e-03 & 2.13e-04 & 3.51e-04 & 1.29 \\
    \midrule
    \textbf{soe\_1} & -     & -     & -     & -     & $>$ 1 day \\
    \midrule
    \textbf{soe\_2} & 7.98e-02 & 1.79e-03 & 1.83e-02 & 5.44e-02 & 18.65 \\
    \bottomrule
    \end{tabular}%
    }
  \label{tab:RPMADCIB-mix}%
\end{table}
Finally, on one hand, the two logarithmic formulations essentially lead to the same level of accuracy. On the other hand, \textbf{log\_obj} seems to be the most efficient method for the three datasets considered.

\subsubsection{Comparison of some portfolio selection approaches} \label{sec:insample}

To better illustrate the concept of risk allocation,
 we report here a graphical comparison of
 the five portfolio selection strategies listed below.
 \begin{enumerate}
\item MAD-RP: the portfolio obtained by solving
the \textbf{log\_obj} formulation.
% (since this formulation appears to be sufficiently accurate and efficient).
%
\item MinMAD: the global minimum MAD portfolio, obtained by implementing problem (3.6) of \cite{Konno91} without constraining the portfolio expected return.
\item MinV: the global minimum variance portfolio, as in \cite{Mark:52}.
\item Vol-RP: the portfolio obtained by means of the Risk Parity approach using volatility as a risk measure \citep[Problem (7) of][]{Maillard2010}.
\item EW: the Equally Weighted portfolio, where the capital is uniformly distributed among all the assets in the investment universe, i.e., $x_{i}=\frac{1}{n}$ for each $i\in \{1,\ldots,n\}$.
\end{enumerate}
More precisely, we consider an investment universe of 5 assets and we evaluate
weights and Relative Risk Contributions (RRCs), obtained from all the portfolio
strategies analyzed.
Note that for each portfolio approach, RRCs are computed by
considering both volatility and MAD as risk measures.
\begin{figure}
  \centering
  \includegraphics[width=.99\textwidth]{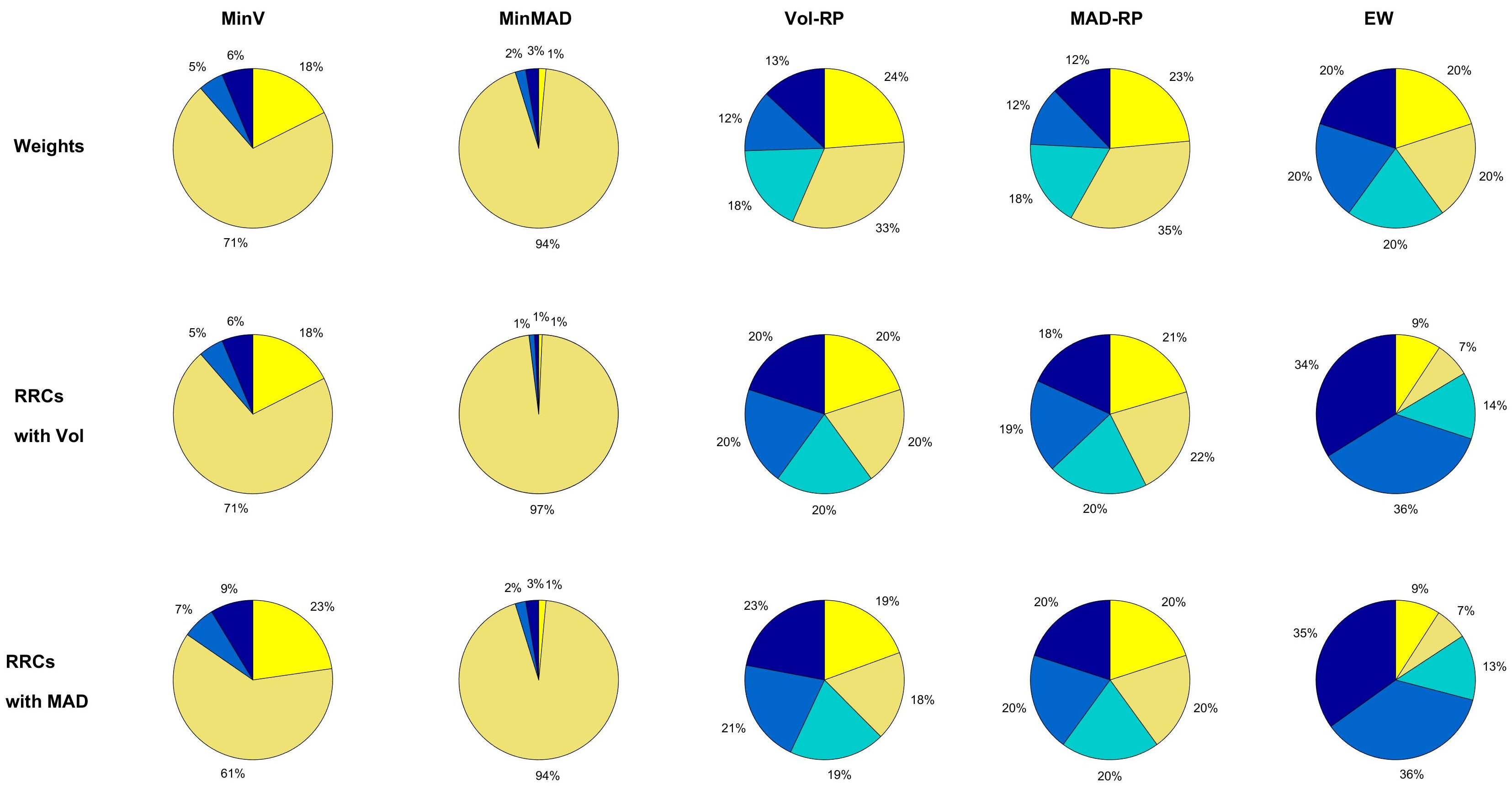}
  \caption{Weights (top), RRCs with volatility (middle), and RRCs with volatility (bottom) pie charts for 5 assets belonging to CIB-mix.}\label{fig:ISWeightsTRCCIB}
\end{figure}

 In Figure~\ref{fig:ISWeightsTRCCIB}, we report
 the results obtained for 5 assets belonging to CIB-mix\footnote{IT Benchmark 10 years DS Govt. Index, IT Benchmark 5 years DS Govt. Index, IT Benchmark 15 years DS Govt. Index, S\&P GSCI Agriculture Total Return, and S\&P GSCI Industrial Metals Total Return for the period 1/2013-12/2018.}.
 We can observe that the Vol-RP and MAD-RP portfolios generally show
 a slightly different portfolio composition but both appear to be well-diversified
 similar to the EW portfolio.
 A converse situation can be pointed out for
the minimum risk portfolios.
Indeed, as shown on the top of Figure \ref{fig:ISWeightsTRCCIB},
both the MinV and MinMAD approaches
concentrate 71\% and 94\% of the capital
in only one asset, respectively.
%
%tend to invest more than 60\% of the capital in only 2 assets. This behaviour is even more evident for CIB-mix (see the left side of Figure \ref{fig:ISWeightsTRCCIB}), where MinV and MinMAD concentrate 71\% and 94\% of the capital
%in only one asset.
%
%On the right side of Figures \ref{fig:ISWeightsTRCETFEC} and \ref{fig:ISWeightsTRCCIB},
%we can see how much each asset contributes (in percentage) to the whole portfolio risk.
%
%
In the middle (bottom) of Figure \ref{fig:ISWeightsTRCCIB}, for each portfolio strategy,
we can see how much each asset contributes (in percentage) to the whole portfolio volatility (MAD).
Clearly, for Vol-RP (MAD-RP), each asset equally contributes to the portfolio volatility (MAD).
Furthermore, by construction, for a fixed risk measure, the asset weights of the minimum risk portfolio coincide with their RRCs.
It is also interesting to observe that, for the EW portfolio, RRCs are uneven regardless of the risk measure used.

  \subsubsection{Out-of-sample analysis} \label{sec:out-of-sample}

In this section, we discuss the main results of the out-of-sample performance
analysis for all the portfolio selection strategies listed in Section \ref{sec:insample}.
For this analysis, we adopt a Rolling Time Window scheme of evaluation:
we allow for the possibility of rebalancing the portfolio composition during the holding period, at fixed intervals. To calibrate the portfolio selection models,
we use 2 years (i.e., 500 days) for the in-sample window,
while we choose 20 days for the out-of-sample window
with rebalancing allowed every 20 days.	
To evaluate the out-of-sample performances of the five portfolios analyzed,
we use the following quantities typically adopted in the literature \citep[see, e.g.,][]{Bacon2008,Rachev2008,bruni2017exact,cesarone2015linear,cesarone2016optimally,cesarone2019risk,cesarone2022does}.
Let $R^{out}$ denote the out-of-sample portfolio return, and $W_{t}=W_{t-1}(1+R^{out})$
the portfolio wealth at time $t$.
\begin{itemize}
\item \textbf{Mean} is the annualized average portfolio return, i.e., $\mu^{out}=\E[R^{out}]$. Obviously, higher values are associated to higher performances.
\item Volatility (\textbf{Vol}) is the annualized standard deviation, computed as $\sigma^{out}=\sqrt{\E[(R^{out}-\mu^{out})^{2}]}$. Since it measures the portfolio risk, lower values are preferable.
\item Mean Absolute Deviation (\textbf{MAD}) of $R^{out}$, namely
$\MAD(R^{out}):=\E[\abs{R^{out}-\mu^{out}}]$.
\item Maximum Drawdown (\textbf{MaxDD}): denoting the \textit{drawdowns} by
\begin{equation*}
  dd_{t}=\frac{W_{t}-\displaystyle\max_{1\leq \tau\leq t}W_{\tau}}{\displaystyle\max_{1\leq \tau\leq t}W_{\tau}}\mbox{,}\quad t\in\{1,\ldots,T\},
\end{equation*}
%  											
%  											where we set $W_0=1$.
\textbf{MaxDD} is defined as
$
MaxDD=\displaystyle\min_{1\leq t\leq T}dd_{t}\mbox{.}
$

\noindent
It measures the loss from the observed peak in the returns; hence, it will always have negative sign (or, in the best case scenario, it will be equal to 0), meaning that the higher the better.
\item \textbf{Ulcer} index, i.e.,
\begin{equation*}
UI=\sqrt{\frac{\sum\limits_{t=1}^{T}dd_{t}^{2}}{T}}\mbox{.}
\end{equation*}
It evaluates the depth and the duration of drawdowns in prices over the out-of-sample period.
Lower \textbf{Ulcer} values are associated to better portfolio performances.
\item \textbf{Sharpe} ratio measures the gain per unit risk and it is defined as
\begin{equation*}
SR^{out}=\frac{\mu^{out}-r_{f}}{\sigma^{out}},
\end{equation*}
where we choose $r_{f}=0$. The larger this value, the better the performance.
\item \textbf{Sortino} ratio is similar to the Sharpe ratio but uses another risk measure, i.e., the Target Downside Deviation, $TDD=\sqrt{\E[((R^{out}-r_{f})^-)^{2}]}$. Therefore, the Sortino Ratio is defined as
  											\begin{equation*}
  											SoR=\frac{\mu^{out}-r_{f}}{TDD},
  											\end{equation*}
  											where $r_{f}=0$. Similar to the Sharpe ratio, the higher it is, the better the portfolio performance.
  											\item \textbf{Turnover}, i.e.,
  											\begin{equation*}
  											Turn=\frac{1}{Q}\sum\limits_{q=1}^{Q}\sum\limits_{k=1}^{n}|x_{q,k}-x_{q-1,k}|,
  											\end{equation*}
  %
  %where $Q$ represents the number of rebalances. Since, typically, rebalancing the portfolio is not without cost, lower Turnover values indicate a better portfolio performance.
  where $Q$ represents the number of rebalances, $x_{q,k}$ is the portfolio weight of asset $k$ after rebalancing, and $x_{q-1,k}$ is the portfolio weight before rebalancing at time $q$. Lower turnover values imply better portfolio performance.
  We point out that this definition of portfolio turnover is a proxy of the effective one, since it evaluates only the amount
of trading generated by the models at each rebalance,
without considering the trades due to changes in asset
prices between one rebalance and the next. Thus, by definition,
the turnover of the EW portfolio is zero.
\item \textbf{Rachev-5\%} ratio measures the upside potential, comparing right and left tail. Mathematically, it is computed as
  											\begin{equation*}
  											\frac{CVaR_{\alpha}(r_{f}-R^{out})}{CVaR_{\beta}(R^{out}-r_{f})},
  											\end{equation*}
  											where we choose $\alpha=\beta = 5\%$ and $r_{f}=0$.
  										\end{itemize}
 			
In Tables \ref{tab:OSETF500}, \ref{tab:OSEuroBonds500} and \ref{tab:OSItBondsandCommodities500}, we report the experimental results obtained for the three dataset, ETF-EC, EuroBonds and CIB-mix, respectively.
For each performance measure, we show with different colors the rank of the results of the five portfolio selection strategies analyzed.
More precisely, the colors range
from deep-green to deep-red, where deep-green represents
the best performance while deep-red the worst
one.
% Table generated by Excel2LaTeX from sheet 'Table 5'
\begin{table}[htbp]
  \centering
  \caption{Out-of-sample results for ETF-EC.}
  \scalebox{0.89}{
    \begin{tabular}{|l|c|c|c|c|c|}
    \toprule
          & \textbf{MinV} & \textbf{MinMAD} & \textbf{Vol-RP} & \textbf{MAD-RP} & \textbf{EW} \\
    \midrule
    \textbf{Mean} & \cellcolor[rgb]{ 1,  0,  0}0.0485 & \cellcolor[rgb]{ 1,  .753,  0}0.0605 & \cellcolor[rgb]{ .573,  .816,  .314}0.0666 & \cellcolor[rgb]{ .573,  .816,  .314}0.0666 & \cellcolor[rgb]{ 0,  .69,  .314}0.0703 \\
    \midrule
    \textbf{Vol} & \cellcolor[rgb]{ .573,  .816,  .314}0.0816 & \cellcolor[rgb]{ 0,  .69,  .314}0.0801 & \cellcolor[rgb]{ 1,  .753,  0}0.0907 & \cellcolor[rgb]{ 1,  1,  0}0.0904 & \cellcolor[rgb]{ 1,  0,  0}0.1027 \\
    \midrule
    \textbf{MAD} & \cellcolor[rgb]{ .573,  .816,  .314}1.0087 & \cellcolor[rgb]{ 0,  .69,  .314}0.9777 & \cellcolor[rgb]{ 1,  .753,  0}1.0990 & \cellcolor[rgb]{ 1,  1,  0}1.0941 & \cellcolor[rgb]{ 1,  0,  0}1.2485 \\
    \midrule
    \textbf{MaxDD} & \cellcolor[rgb]{ .573,  .816,  .314}-0.1532 & \cellcolor[rgb]{ 0,  .69,  .314}-0.1514 & \cellcolor[rgb]{ 1,  1,  0}-0.1959 & \cellcolor[rgb]{ 1,  .753,  0}-0.1989 & \cellcolor[rgb]{ 1,  0,  0}-0.2218 \\
    \midrule
    \textbf{Ulcer} & \cellcolor[rgb]{ .573,  .816,  .314}0.0627 & \cellcolor[rgb]{ 0,  .69,  .314}0.0601 & \cellcolor[rgb]{ 1,  1,  0}0.0819 & \cellcolor[rgb]{ 1,  .753,  0}0.0838 & \cellcolor[rgb]{ 1,  0,  0}0.0960 \\
    \midrule
    \textbf{Sharpe} & \cellcolor[rgb]{ 1,  0,  0}0.0376 & \cellcolor[rgb]{ 0,  .69,  .314}0.0477 & \cellcolor[rgb]{ 1,  1,  0}0.0465 & \cellcolor[rgb]{ .573,  .816,  .314}0.0466 & \cellcolor[rgb]{ 1,  .753,  0}0.0433 \\
    \midrule
    \textbf{Sortino} & \cellcolor[rgb]{ 1,  0,  0}0.0527 & \cellcolor[rgb]{ 0,  .69,  .314}0.0671 & \cellcolor[rgb]{ 1,  1,  0}0.0641 & \cellcolor[rgb]{ .573,  .816,  .314}0.0643 & \cellcolor[rgb]{ 1,  .753,  0}0.0597 \\
    \midrule
    \textbf{Turnover} & \cellcolor[rgb]{ 1,  .753,  0}0.1184 & \cellcolor[rgb]{ 1,  0,  0}0.1477 & \cellcolor[rgb]{ 0,  .69,  .314}0.0242 & \cellcolor[rgb]{ 1,  1,  0}0.0329 & 0 \\
    \midrule
    \textbf{Rachev-5\%} & \cellcolor[rgb]{ .573,  .816,  .314}0.8473 & \cellcolor[rgb]{ 0,  .69,  .314}0.8493 & \cellcolor[rgb]{ 1,  1,  0}0.8381 & \cellcolor[rgb]{ 1,  .753,  0}0.8360 & \cellcolor[rgb]{ 1,  0,  0}0.8270 \\
    \bottomrule
    \end{tabular}%
    }
  \label{tab:OSETF500}%
\end{table}%								

As shown in Table \ref{tab:OSETF500}, the MinMAD portfolio seems to have better performance with respect to the other portfolios with two exceptions: \textbf{Turnover} and \textbf{Mean}.
Indeed, \textbf{Turnover} of the Vol-RP and MAD-RP portfolios are much smaller than that of MinMAD.
The EW portfolio guarantees the highest \textbf{Mean}, followed by the Vol-RP and MAD-RP portfolios.
We can observe that these two RP strategies provide very similar values
for all the performance measures considered,
and their out-of-sample performances, both in terms of risk and of gain, are located in the middle between those of the minimum risk portfolios and those of the EW portfolio, as theoretically expected in the in-sample case (see Remarks \ref{rem:MinMADvsMAD-ERC} and \ref{rem:EWvsMAD-ERC}).
The MinV and MinMAD portfolios appear to be the least risky, both in terms of \textbf{Vol}, \textbf{MAD}, \textbf{MaxDD}, and \textbf{Ulcer}.
Furthermore, even though MAD-RP tends to be riskier than the minimum risk strategies,
it achieves the second-largest values of \textbf{Sharpe} and \textbf{Sortino}.
Note that we do purposely leave the EW portfolio out of the comparison concerning \textbf{Turnover}, since its value is 0 by construction.
%
 % Table generated by Excel2LaTeX from sheet 'Table 6'
\begin{table}[htbp]
  \centering
  \caption{Out-of-sample results for EuroBonds.}
  \scalebox{0.89}{
    \begin{tabular}{|l|c|c|c|c|c|}
    \toprule
          & \textbf{MinV} & \textbf{MinMAD} & \textbf{Vol-RP} & \textbf{MAD-RP} & \textbf{EW} \\
    \midrule
    \textbf{Mean} & \cellcolor[rgb]{ 1,  .753,  0}-0.0004 & \cellcolor[rgb]{ 1,  0,  0}-0.0006 & \cellcolor[rgb]{ 1,  1,  0}0.0082 & \cellcolor[rgb]{ .573,  .816,  .314}0.0085 & \cellcolor[rgb]{ 0,  .69,  .314}0.0272 \\
    \midrule
    \textbf{Vol} & \cellcolor[rgb]{ .573,  .816,  .314}0.0046 & \cellcolor[rgb]{ 0,  .69,  .314}0.0043 & \cellcolor[rgb]{ 1,  1,  0}0.0117 & \cellcolor[rgb]{ 1,  .753,  0}0.0120 & \cellcolor[rgb]{ 1,  0,  0}0.0379 \\
    \midrule
    \textbf{MAD} & \cellcolor[rgb]{ .573,  .816,  .314}0.0442 & \cellcolor[rgb]{ 0,  .69,  .314}0.0425 & \cellcolor[rgb]{ 1,  1,  0}0.1323 & \cellcolor[rgb]{ 1,  .753,  0}0.1346 & \cellcolor[rgb]{ 1,  0,  0}0.4334 \\
    \midrule
    \textbf{MaxDD} & \cellcolor[rgb]{ .573,  .816,  .314}-0.0122 & \cellcolor[rgb]{ 0,  .69,  .314}-0.0119 & \cellcolor[rgb]{ 1,  1,  0}-0.0204 & \cellcolor[rgb]{ 1,  .753,  0}-0.0206 & \cellcolor[rgb]{ 1,  0,  0}-0.0716 \\
    \midrule
    \textbf{Ulcer} & \cellcolor[rgb]{ 0,  .69,  .314}0.0047 & \cellcolor[rgb]{ .573,  .816,  .314}0.0049 & \cellcolor[rgb]{ 1,  .753,  0}0.0064 & \cellcolor[rgb]{ 1,  1,  0}0.0062 & \cellcolor[rgb]{ 1,  0,  0}0.0223 \\
    \midrule
    \textbf{Sharpe} & \cellcolor[rgb]{ 1,  .753,  0}- & \cellcolor[rgb]{ 1,  0,  0}- & \cellcolor[rgb]{ 1,  1,  0}0.0443 & \cellcolor[rgb]{ .573,  .816,  .314}0.0447 & \cellcolor[rgb]{ 0,  .69,  .314}0.0454 \\
    \midrule
    \textbf{Sortino} & \cellcolor[rgb]{ 1,  .753,  0}- & \cellcolor[rgb]{ 1,  0,  0}- & \cellcolor[rgb]{ 1,  1,  0}0.0604 & \cellcolor[rgb]{ .573,  .816,  .314}0.0611 & \cellcolor[rgb]{ 0,  .69,  .314}0.0631 \\
    \midrule
    \textbf{Turnover} & \cellcolor[rgb]{ 1,  .753,  0}0.1119 & \cellcolor[rgb]{ 1,  0,  0}0.1153 & \cellcolor[rgb]{ 0,  .69,  .314}0.0199 & \cellcolor[rgb]{ 1,  1,  0}0.0401 & 0 \\
    \midrule
    \textbf{Rachev-5\%} & \cellcolor[rgb]{ 1,  0,  0}0.8212 & \cellcolor[rgb]{ 1,  .753,  0}0.8297 & \cellcolor[rgb]{ 1,  1,  0}0.8387 & \cellcolor[rgb]{ .573,  .816,  .314}0.8465 & \cellcolor[rgb]{ 0,  .69,  .314}0.9054 \\
    \bottomrule
    \end{tabular}%
    }
  \label{tab:OSEuroBonds500}%
\end{table}%

In Table \ref{tab:OSEuroBonds500}, we report the computational results for EuroBonds.
%on which we can make three distinct comments.
%; those concerning the risk, those concerning remuneration and those concerning rebalancing.
%  										As for the first ones,
Again as expected, the minimum risk portfolios show the best risk performances but with the worst values of out-of-sample expected return.
% perform better than the others: especially the MinMAD portfolio yields a very good performance in terms of risk.
%
Conversely, the EW portfolio has the best values in terms of Gain-Risk ratios and the worst in terms of risk.
Also in this case, the performances of the RP strategies generally lie between those of EW and those of the minimum risk portfolios.
%; their risk is higher than that of both MinV and MinMAD, but lower than that of EW, and viceversa for return and risk remuneration.
%
Furthermore, the Vol-RP portfolio shows the lowest value of \textbf{Turnover}, followed by MAD-RP.
%
  % Table generated by Excel2LaTeX from sheet 'Table 7'
\begin{table}[htbp]
  \centering
  \caption{Out-of-sample results for CIB-mix.}
  \scalebox{0.89}{
    \begin{tabular}{|l|c|c|c|c|c|}
    \toprule
          & \textbf{MinV} & \textbf{MinMAD} & \textbf{Vol-RP} & \textbf{MAD-RP} & \textbf{EW} \\
          \midrule
    \textbf{Mean} & \cellcolor[rgb]{ 1,  .753,  0}0.0025 & \cellcolor[rgb]{ 1,  0,  0}0.0023 & \cellcolor[rgb]{ .573,  .816,  .314}0.0042 & \cellcolor[rgb]{ 0,  .69,  .314}0.0073 & \cellcolor[rgb]{ 1,  1,  0}0.0028 \\
    \midrule
    \textbf{Vol} & \cellcolor[rgb]{ 0,  .69,  .314}0.0201 & \cellcolor[rgb]{ .573,  .816,  .314}0.0208 & \cellcolor[rgb]{ 1,  1,  0}0.0290 & \cellcolor[rgb]{ 1,  .753,  0}0.0298 & \cellcolor[rgb]{ 1,  0,  0}0.0605 \\
    \midrule
    \textbf{MAD} & \cellcolor[rgb]{ 0,  .69,  .314}0.1110 & \cellcolor[rgb]{ .573,  .816,  .314}0.1198 & \cellcolor[rgb]{ 1,  1,  0}0.3125 & \cellcolor[rgb]{ 1,  .753,  0}0.3296 & \cellcolor[rgb]{ 1,  0,  0}0.7194 \\
    \midrule
    \textbf{MaxDD} & \cellcolor[rgb]{ 0,  .69,  .314}-0.0488 & \cellcolor[rgb]{ .573,  .816,  .314}-0.0498 & \cellcolor[rgb]{ 1,  1,  0}-0.0541 & \cellcolor[rgb]{ 1,  .753,  0}-0.0559 & \cellcolor[rgb]{ 1,  0,  0}-0.1335 \\
    \midrule
    \textbf{Ulcer} & \cellcolor[rgb]{ 0,  .69,  .314}0.0080 & \cellcolor[rgb]{ .573,  .816,  .314}0.0086 & \cellcolor[rgb]{ 1,  .753,  0}0.0261 & \cellcolor[rgb]{ 1,  1,  0}0.0244 & \cellcolor[rgb]{ 1,  0,  0}0.0657 \\
    \midrule
    \textbf{Sharpe} & \cellcolor[rgb]{ 1,  1,  0}0.0079 & \cellcolor[rgb]{ 1,  .753,  0}0.0069 & \cellcolor[rgb]{ .573,  .816,  .314}0.0092 & \cellcolor[rgb]{ 0,  .69,  .314}0.0155 & \cellcolor[rgb]{ 1,  0,  0}0.0030 \\
    \midrule
    \textbf{Sortino} & \cellcolor[rgb]{ 1,  1,  0}0.0097 & \cellcolor[rgb]{ 1,  .753,  0}0.0084 & \cellcolor[rgb]{ .573,  .816,  .314}0.0125 & \cellcolor[rgb]{ 0,  .69,  .314}0.0214 & \cellcolor[rgb]{ 1,  0,  0}0.0042 \\
    \midrule
    \textbf{Turnover} & \cellcolor[rgb]{ 1,  .753,  0}0.0471 & \cellcolor[rgb]{ 1,  1,  0}0.0382 & \cellcolor[rgb]{ 0,  .69,  .314}0.0288 & \cellcolor[rgb]{ 1,  0,  0}0.0886 & 0 \\
    \midrule
    \textbf{Rachev-5\%} & \cellcolor[rgb]{ 1,  0,  0}0.8123 & \cellcolor[rgb]{ 1,  .753,  0}0.8326 & \cellcolor[rgb]{ 1,  1,  0}0.9074 & \cellcolor[rgb]{ .573,  .816,  .314}0.9217 & \cellcolor[rgb]{ 0,  .69,  .314}0.9924 \\
    \bottomrule
    \end{tabular}%
    }
  \label{tab:OSItBondsandCommodities500}%
\end{table}%

In Table \ref{tab:OSItBondsandCommodities500}, we report the computational results for CIB-mix.
We observe here that the MAD-RP portfolio presents very good performances:
it has the highest \textbf{Mean}, \textbf{Sharpe} and \textbf{Sortino} values (followed by the Vol-RP portfolio), and the second-highest \textbf{Rachev-5\%}, below that of EW.
Again, the minimum risk portfolios, particularly MinV, have the best \textbf{Vol}, \textbf{MAD}, \textbf{MaxDD}, and \textbf{Ulcer}. Furthermore, the Vol-RP portfolio provides the lowest \textbf{Turnover}, followed by MinMAD.

In Figure \ref{fig:W500}, for ETF-EC, EuroBonds and CIB-mix, we report the time evolution of the portfolio wealth, investing one unit of the currency at the beginning of the investment horizon.
%namely
%how much money one can expect to have if they invest one unit of the currency at the beginning of the investment horizon.
%
We can note that the EW portfolio often provides the highest values of wealth and that the Risk Parity portfolios are systematically better than the minimum risk ones.
The Vol-RP and MAD-RP portfolios tend to perform similarly, with the only exception of the CIB-mix dataset, where the MAD-RP portfolio performs significantly better and is often the best among all.
  										
  										\begin{figure}[h!]
  											\centering
  											\includegraphics[width=.47\textwidth]{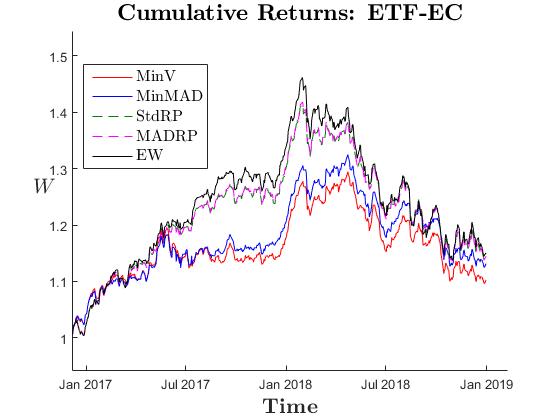}\quad
  											\includegraphics[width=.47\textwidth]{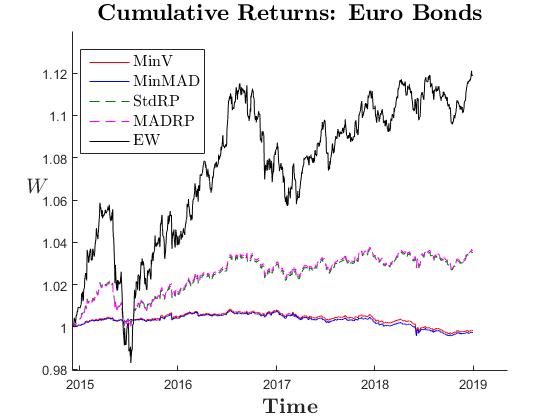}
  											
  											\medskip
  											
  											\includegraphics[width=.47\textwidth]{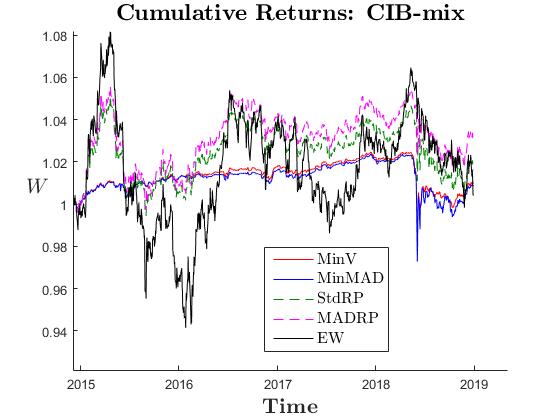}
  											
  											\caption{Time evolution of the portfolio wealth for ETF-EC, EuroBonds and CIB-mix.}
  											\label{fig:W500}
  										\end{figure}
%

\begin{table}[htbp]
  \centering
  \caption{Analysis of the differences in the portfolio weights of the Risk Parity and minimum risk strategies.
}
    \begin{tabular}{|c|c|c|c|}
    \toprule
          & \multicolumn{1}{l|}{\textbf{ETF-EC}} & \multicolumn{1}{l|}{\textbf{EuroBonds}} & \multicolumn{1}{l|}{\textbf{CIB-mix}} \\
    \midrule
    \textbf{Vol-RP  vs. MAD-RP} & 0.0023 & 0.0020 & 0.0197 \\
    \midrule
    \textbf{MinV vs. MinMAD} & 0.0097 & 0.0038 & 0.0154 \\
    \bottomrule
    \end{tabular}%
  \label{tab:WeightDifferences}%
\end{table}%
Since the out-of-sample performance measures for the Risk Parity portfolios using volatility and MAD generally tend to appear very close, we compute a simple metric to better highlight the differences between the portfolio weights selected by these two approaches. More precisely, for each dataset listed in Table \ref{tab:Datasets}, we use the following metric named Overall Weight Differences (OWD):
\begin{equation}
  OWD =\frac{1}{Q} \sum\limits_{q=1}^{Q} \left( \frac{1}{n} \sum\limits_{k=1}^{n} |x_{q,k}^{Vol-RP}-x_{q,k}^{MAD-RP}| \right),
\end{equation}
where, again, $Q$ represents the number of rebalances, $x_{q,k}^{Vol-RP}$ and $x_{q,k}^{MAD-RP}$ are the weights of asset $k$ obtained, at time $q$, by the RP strategy using volatility and MAD, respectively.
Table \ref{tab:WeightDifferences} shows the results obtained for all datasets considered, in which the analysis of the differences in the portfolio weights of the MinV and MinMAD approaches is also included.
For each dataset, we can observe that the OWD values computed for the Risk Parity and minimum risk strategies appear to be of the same order of magnitude.
Furthermore, as highlighted for the out-of-sample performance results, it seems that the greatest portfolio allocation differences can be observed for the CIB-mix dataset.									
  										
\section{Conclusions} \label{sec:ConclusionsMADRP}

In this paper, we have proposed an extension of the now popular Risk Parity approach for volatility to the Mean Absolute Deviation (MAD) as a portfolio risk measure.
%

%\noindent
From a theoretical viewpoint, we have discussed the subdifferentiability and additivity features of MAD, highlighting their usefulness in the Risk Parity strategy,
and we have established
the existence and uniqueness results for the MAD-RP portfolio.
Furthermore, we have presented several mathematical formulations for finding the MAD-RP portfolios practically and we have tested them on three real-world datasets both in terms of efficiency and accuracy.
We have observed that the system-of-equation formulations seem to be the least accurate and efficient. Methods based on least-squares formulations show intermediate performances for all the datasets, while those based on logarithmic formulations appear to be the best.
%

%\noindent
From an experimental viewpoint,
we have presented an extensive empirical analysis
comparing the out-of-sample performance obtained from the global minimum volatility and MAD
portfolios, the volatility and MAD Risk Parity portfolios, and the Equally Weighted
portfolio.
In terms of out-of-sample risk, these results confirm the theoretical (in-sample) properties of the RP portfolio, namely its risk lies between that of the minimum risk portfolio and the risk of the Equally Weighted portfolio.
The RP portfolios always show positive annual expected returns for all datasets, even though they are smaller than those of the Equally Weighted portfolio, which is however much riskier.
Furthermore, in terms of Gain-Risk ratios, the
Risk Parity portfolios tend to provide medium-high performances, often between those of EW and those of minimum risk portfolios; however, in few cases, they are better.
We point out that the two Risk Parity portfolios show similar values for all performance measures.
However, when directly comparing Vol-RP with MAD-RP, the latter usually provides slightly higher profitability at the expense of slightly higher risk and turnover.
Furthermore, for each dataset, we have observed that the portfolio allocation differences computed through the Risk Parity and minimum risk strategies using volatility and MAD seem to be of the same order of magnitude.

%\noindent
Each model tested tends to respond to different requirements related to different risk attitudes of the investors.
On one hand, as expected, the
minimum risk models are advisable for risk-averse
investors, avoiding as much as possible any shock
represented by deep drawdowns.
On the other hand, the RP strategies seem to be appropriate for investors mildly adverse to the total portfolio risk. These investors could be willing to waive a bit of safety (w.r.t. that of the minimum risk portfolios) and a bit of gain (w.r.t. that of the EW portfolio) to achieve a more balanced risk allocation and a more diversified portfolio.
Furthermore, although the Equally Weighted approach embodies the concept of high diversification and achieves good out-of-sample expected returns, it generates portfolios with very high out-of-sample risk both in terms of volatility, MAD, Maximum Drawdown, and Ulcer index. Therefore, according to our findings, the EW portfolio seems to be advisable for sufficiently risk-seeking investors who try to maximize gain without worrying about periods of deep drawdowns.

\appendix
%%%%%%%%%%%%%%%%%%%%%%%%
\section{Proofs of the results in Section \ref{sec:add}}\label{app:Lemma}
\begin{proof}[Proof of Lemma \ref{MADAdditive}]
By the triangle inequality on $\mathbb{R}$, we have that
\begin{equation}\label{eq:add0}
 \abs{X-\E[X]} + \abs{Y-\E[Y]} -
 \abs{(X-\E[X]) + (Y-\E[Y])} \geq 0  \quad \mathbb{P}\mbox{-a.s.}
\end{equation}
Note that the additivity condition \eqref{eq:add} is equivalent to
\begin{equation}\label{eq:add1}
  \E\sqb{\abs{ (X-\E[X])+(Y-\E[Y])} -
 \abs{X-\E[X]} - \abs{Y-\E[Y]}} = 0
\end{equation}
By \eqref{eq:add0}, the random variable $\abs{X-\E[X]} + \abs{Y-\E[Y]} -
\abs{(X-\E[X]) + (Y-\E[Y])} $ is $\P$-a.s. nonnegative. Hence, \eqref{eq:add1} is equivalent to
  \begin{equation}\label{eq:add2}
  \abs{(X-\E[X]) + (Y-\E[Y])} - \abs{X-\E[X]} - \abs{Y-\E[Y]}
    = 0   \quad \P\mbox{-a.s.}
  \end{equation}
Now, for each $a,b\in \mathbb{R}$, we have $\lvert a+b \rvert = \lvert a\rvert+\lvert b\rvert$
if and only if $ab\geq 0$. It follows that \eqref{eq:add2} is equivalent to \eqref{eq:add3}, which completes the proof.
\end{proof}	

%%%%%%%%%%%%%%%%%%%%%%%%%
%\section{Proof of Theorem \ref{th:AddGen}} 	\label{app:Theorem3}
%			

\begin{proof}[Proof of Theorem \ref{th:AddGen}]
	It is clear that (i) implies (ii) and (iii). Next, we show that (ii) implies (i). Let $\boldsymbol{x}\in\R^n_+$. If $\boldsymbol{x}=0$, then the equality in (i) becomes trivial. Suppose that $\boldsymbol{x}\neq 0$. In particular, $a:=\sum_{i=1}^n x_i>0$. Then, by (ii) and the positive homogeneity of MAD, we have
	\[
	\MAD\of{\sum_{i=1}^nx_i Y_i}=a\MAD\of{\sum_{i=1}^n \frac{x_i}{a}Y_i}=a\sum_{i=1}^n \frac{x_i}{a}\MAD(Y_i)=\sum_{i=1}^n x_i\MAD(Y_i).
	\]
	Hence, (i) follows.
	
	To prove that (iii) implies (iv), let us fix $i,j\in\{1,\ldots,n\}$ with $i\neq j$ and set $I:=\{i,j\}$. Then, (iii) yields $\MAD(Y_i+Y_j)=\MAD(Y_i)+\MAD(Y_j)$. Hence, by Theorem \ref{MADAdditive}, we have $(Y_i-\E[Y_i])(Y_j-\E[Y_j])\geq 0$ $\P$-a.s., that is, (iv) holds.
	
	To complete the proof, we show that (iv) implies (i). We prove the equality in (i) by induction on $k(\boldsymbol{x}):=\min\{i\geq 0\mid x_{i+1}=\ldots=x_n=0\}$. (We assume that $k(\boldsymbol{x})=n$ if $x_n\neq 0$.) Note that $0\leq k(\boldsymbol{x})\leq n$. The base case $k(\boldsymbol{x})=0$ is trivial since $x_1=\ldots=x_n=0$ in this case. Let $k\leq n-1$ and suppose that the equality in (i) holds for every $\bar{\boldsymbol{x}}\in\R^n_+$ with $k(\bar{\boldsymbol{x}})=k$. Let $\boldsymbol{x}\in\R^n_+$ with $k(\boldsymbol{x})=k+1$. Then,
	\begin{equation}\label{eq:multivar1}
	\MAD\of{\sum_{i=1}^n x_iY_i}=\MAD\of{\sum_{i=1}^{k+1} x_iY_i}=\MAD\of{\sum_{i=1}^{k} x_iY_i+x_{k+1}Y_{k+1}}.
	\end{equation}
	Note that we have
	\begin{align*}
	&\of{\sum_{i=1}^{k} x_iY_i-\E\sqb{\sum_{i=1}^{k} x_iY_i}}\of{x_{k+1}Y_{k+1}-x_{k+1}\E[Y_{k+1}]}\\
	&=x_{k+1}\sum_{i=1}^{k}x_i(Y_i-\E[Y_i])(Y_{k+1}-\E[Y_{k+1}])\geq 0
	\end{align*}
	thanks to (iv) and the fact that $\boldsymbol{x}\in\R^n_+$. Hence, by Theorem \ref{MADAdditive},
	\begin{equation}\label{eq:multivar2}
	\MAD\of{\sum_{i=1}^{k} x_iY_i+x_{k+1}Y_{k+1}}=\MAD\of{\sum_{i=1}^{k} x_iY_i}+\MAD(x_{k+1}Y_{k+1}).
	\end{equation}
	On the other hand, we may write $\sum_{i=1}^{k} x_iY_i=\sum_{i=1}^{n} \bar{x}_iY_i$, where $\bar{\boldsymbol{x}}:=(x_1,\ldots,x_k,0,\ldots,0)$. Since $k(\bar{
		\boldsymbol{x}})=k$, we may use the induction hypothesis and obtain
	\begin{equation}\label{eq:multivar3}
	\MAD\of{\sum_{i=1}^{k} x_iY_i}=\MAD\of{\sum_{i=1}^{n} \bar{x}_iY_i}=\sum_{i=1}^n \bar{x}_i\MAD(Y_i)=\sum_{i=1}^k x_i\MAD(Y_i).
	\end{equation}
	Combining \eqref{eq:multivar1}, \eqref{eq:multivar2}, \eqref{eq:multivar3} and using the positive homogeneity of MAD give that
	\[
	\MAD\of{\sum_{i=1}^n x_iY_i}=\sum_{i=1}^k x_i\MAD(Y_i)+x_{k+1}\MAD(Y_{k+1})=\sum_{i=1}^n x_i\MAD(Y_i),
	\]
	which concludes the induction argument.
	\end{proof}					

\section*{Conflict of interest}

We declare that we do not have any conflicting interests related to this article.

\section*{Acknowledgments}

This work was partially supported by the PRIN 2017 Project (no. 20177WC4KE), funded by the Italian Ministry of Education, University, and Research. We thank two anonymous referees for their useful comments. In particular, Example~\ref{ex1} is suggested by one of the referees.

{\footnotesize
\bibliographystyle{spbasic}
\bibliography{BibbaseMADRP20240117.bib}

\begin{thebibliography}{46}
\providecommand{\natexlab}[1]{#1}
\providecommand{\url}[1]{{#1}}
\providecommand{\urlprefix}{URL }
\expandafter\ifx\csname urlstyle\endcsname\relax
  \providecommand{\doi}[1]{DOI~\discretionary{}{}{}#1}\else
  \providecommand{\doi}{DOI~\discretionary{}{}{}\begingroup
  \urlstyle{rm}\Url}\fi
\providecommand{\eprint}[2][]{\url{#2}}

\bibitem[{Bacon(2008)}]{Bacon2008}
Bacon CA (2008) Practical Portfolio Performance Measurement and Attribution,
  2nd edn. Wiley

\bibitem[{Bai et~al(2016)Bai, Scheinberg, and Tutuncu}]{bai2016least}
Bai X, Scheinberg K, Tutuncu R (2016) Least-squares approach to risk parity in
  portfolio selection. Quantitative Finance 16(3):357--376

\bibitem[{Bellini et~al(2021)Bellini, Cesarone, Colombo, and
  Tardella}]{bellini2021risk}
Bellini F, Cesarone F, Colombo C, Tardella F (2021) Risk parity with
  expectiles. European Journal of Operational Research 291(3):1149--1163

\bibitem[{Best and Grauer(1991{\natexlab{a}})}]{Best1991b}
Best M, Grauer R (1991{\natexlab{a}}) On the sensitivity of
  mean-variance-efficient portfolios to changes in asset means: some analytical
  and computational results. Review of Financial Studies 4(2):315--342

\bibitem[{Best and Grauer(1991{\natexlab{b}})}]{Best1991a}
Best M, Grauer R (1991{\natexlab{b}}) {Sensitivity analysis for mean-variance
  portfolio problems}. Management Science 37(8):980--989

\bibitem[{Boudt et~al(2013)Boudt, Carl, and Peterson}]{Boudt2013}
Boudt K, Carl P, Peterson B (2013) {Asset allocation with Conditional
  Value-at-Risk budgets}. Journal of Risk 15:39--68

\bibitem[{Bruni et~al(2017)Bruni, Cesarone, Scozzari, and
  Tardella}]{bruni2017exact}
Bruni R, Cesarone F, Scozzari A, Tardella F (2017) On exact and approximate
  stochastic dominance strategies for portfolio selection. European Journal of
  Operational Research 259(1):322--329

\bibitem[{Carleo et~al(2017)Carleo, Cesarone, Gheno, and
  Ricci}]{carleo2017approximating}
Carleo A, Cesarone F, Gheno A, Ricci JM (2017) Approximating exact expected
  utility via portfolio efficient frontiers. Decisions in Economics and Finance
  40(1-2):115--143

\bibitem[{Cesarone(2020)}]{cesarone2020computational}
Cesarone F (2020) Computational Finance: MATLAB{\textregistered} Oriented
  Modeling. Routledge

\bibitem[{Cesarone and Colucci(2018)}]{cesarone2018minimum}
Cesarone F, Colucci S (2018) Minimum risk versus capital and risk
  diversification strategies for portfolio construction. Journal of the
  Operational Research Society 69(2):183--200

\bibitem[{Cesarone and Tardella(2017)}]{cesarone2017equal}
Cesarone F, Tardella F (2017) Equal risk bounding is better than risk parity
  for portfolio selection. Journal of Global Optimization 68(2):439--461

\bibitem[{Cesarone et~al(2015)Cesarone, Scozzari, and
  Tardella}]{cesarone2015linear}
Cesarone F, Scozzari A, Tardella F (2015) Linear vs. quadratic portfolio
  selection models with hard real-world constraints. Computational Management
  Science 12(3):345--370

\bibitem[{Cesarone et~al(2016)Cesarone, Moretti, Tardella
  et~al}]{cesarone2016optimally}
Cesarone F, Moretti J, Tardella F, et~al (2016) Optimally chosen small
  portfolios are better than large ones. Economics Bulletin 36(4):1876--1891

\bibitem[{Cesarone et~al(2019)Cesarone, Lampariello, and
  Sagratella}]{cesarone2019risk}
Cesarone F, Lampariello L, Sagratella S (2019) A risk-gain dominance
  maximization approach to enhanced index tracking. Finance Research Letters
  29:231--238

\bibitem[{Cesarone et~al(2020{\natexlab{a}})Cesarone, Mango, Mottura, Ricci,
  and Tardella}]{cesarone2020stability}
Cesarone F, Mango F, Mottura CD, Ricci JM, Tardella F (2020{\natexlab{a}}) On
  the stability of portfolio selection models. Journal of Empirical Finance
  59:210--234

\bibitem[{Cesarone et~al(2020{\natexlab{b}})Cesarone, Scozzari, and
  Tardella}]{Cesarone2019Anoptimization}
Cesarone F, Scozzari A, Tardella F (2020{\natexlab{b}}) An
  optimization--diversification approach to portfolio selection. Journal of
  Global Optimization 76(2):245--265

\bibitem[{Cesarone et~al(2022)Cesarone, Martino, and Carleo}]{cesarone2022does}
Cesarone F, Martino ML, Carleo A (2022) {Does ESG Impact Really Enhance
  Portfolio Profitability?} Sustainability 14(4):2050

\bibitem[{Chopra and Ziemba(1993)}]{Chopra1993}
Chopra V, Ziemba W (1993) The effect of errors in means, variances, and
  covariances on optimal portfolio choice. The Journal of Portfolio Management
  19(2):6--11

\bibitem[{Clarke et~al(2013)Clarke, De~Silva, and Thorley}]{clarke2013risk}
Clarke R, De~Silva H, Thorley S (2013) Risk parity, maximum diversification,
  and minimum variance: An analytic perspective. The Journal of Portfolio
  Management 39(3):39--53

\bibitem[{DeMiguel et~al(2009)DeMiguel, Garlappi, and Uppal}]{Demiguel2009}
DeMiguel V, Garlappi L, Uppal R (2009) {Optimal versus naive diversification:
  How inefficient is the 1/N portfolio strategy?} Review of Financial Studies
  22(5):1915--1953

\bibitem[{Fabozzi et~al(2021)Fabozzi, Simonian, and Fabozzi}]{fabozzi2021risk}
Fabozzi FA, Simonian J, Fabozzi FJ (2021) Risk parity: The democratization of
  risk in asset allocation. The Journal of Portfolio Management 47:41--50

\bibitem[{Fisher et~al(2015)Fisher, Maymin, and Maymin}]{fisher2015risk}
Fisher GS, Maymin PZ, Maymin ZG (2015) Risk parity optimality. The Journal of
  Portfolio Management 41(2):42--56

\bibitem[{Gordan(1873)}]{gordan}
Gordan P (1873) \"{U}ber die aufl\"{o}sung linearer gleichungen mit reelen
  coeﬀicienten (on the solution of linear inequalities with real
  coeﬀicients). Mathematische Annalen 6(1):23--28

\bibitem[{G\"{u}ler(2010)}]{guler}
G\"{u}ler O (2010) Fundamentals of Optimization. Springer

\bibitem[{Haesen et~al(2017)Haesen, Hallerbach, Markwat, and
  Molenaar}]{haesen2017enhancing}
Haesen D, Hallerbach WG, Markwat T, Molenaar R (2017) Enhancing risk parity by
  including views. The Journal of Investing 26(4):53--68

\bibitem[{Jacobsen and Lee(2020)}]{jacobsen2020risk}
Jacobsen B, Lee W (2020) Risk-parity optimality even with negative {S}harpe
  ratio assets. The Journal of Portfolio Management 46(6):110--119

\bibitem[{Konno and Yamazaki(1991)}]{Konno91}
Konno H, Yamazaki H (1991) {Mean-absolute deviation portfolio optimization
  model and its applications to {T}okyo {S}tock {M}arket}. Management Science
  37(5):519--531

\bibitem[{Liu et~al(2020)Liu, Brzenk, and Cheng}]{liu2020}
Liu B, Brzenk P, Cheng T (2020) Indexing risk parity strategies. S\&P Global,
  S\&P Dow Jones Indices, October 2020 Available at
  \url{https://www.spglobal.com/spdji/en/documents/research/research-indexing-risk-parity-strategies.pdf?force_download=true}
  (October 2020)

\bibitem[{Maillard et~al(2010)Maillard, T, and Teiletche}]{Maillard2010}
Maillard S, T, Teiletche J (2010) {The properties of equally weighted risk
  contribution portfolios}. The Journal of Portfolio Management 36(4):60--70

\bibitem[{Markowitz(1952)}]{Mark:52}
Markowitz H (1952) Portfolio selection. The Journal of Finance 7(1):77--91

\bibitem[{Mausser and Romanko(2018)}]{mausser2018long}
Mausser H, Romanko O (2018) Long-only equal risk contribution portfolios for
  {CVaR} under discrete distributions. Quantitative Finance 18(11):1927--1945

\bibitem[{Michaud and Michaud(1998)}]{Michaud1998}
Michaud R, Michaud R (1998) Efficient Asset Management. Harvard Business School
  Press

\bibitem[{Oderda(2015)}]{oderda2015stochastic}
Oderda G (2015) Stochastic portfolio theory optimization and the origin of
  rule-based investing. Quantitative Finance 15(8):1259--1266

\bibitem[{Qian(2011)}]{qian2011risk}
Qian E (2011) Risk parity and diversification. The Journal of Investing
  20(1):119--127

\bibitem[{Qian(2017)}]{Qian2017}
Qian E (2017) {How na\"{\i}ve is na\"{\i}ve risk parity?} Panagora Asset
  Management

\bibitem[{Rachev et~al(2008)Rachev, Stoyanov, and Fabozzi}]{Rachev2008}
Rachev S, Stoyanov S, Fabozzi F (2008) Advanced Stochastic Models, Risk
  Assessment, and Portfolio Optimization: The Ideal Risk, Uncertainty, and
  Performance Measures. Wiley

\bibitem[{Rockafellar and Wets(1982)}]{rockafellarwets1982}
Rockafellar R, Wets RJB (1982) On the interchange of subdifferentiation and
  conditional expectation for convex functionals. Stochastics 7:173--182

\bibitem[{Rockafellar and Wets(1997)}]{rockafellarwetsbook}
Rockafellar R, Wets RJB (1997) Variational Analysis. Springer

\bibitem[{Rockafellar et~al(2006)Rockafellar, Uryasev, and
  Zabarankin}]{Rockafellar2006}
Rockafellar R, Uryasev S, Zabarankin M (2006) {Generalized deviations in risk
  analysis}. Finance and Stochastics 10(1):51--74

\bibitem[{Rockafellar(1970)}]{rockafellar1970convex}
Rockafellar RT (1970) Convex Analysis. Princeton University Press

\bibitem[{Rockafellar(1973)}]{rockafellarintegral}
Rockafellar RT (1973) Conjugate Duality and Optimization. SIAM

\bibitem[{Roncalli(2013)}]{Roncalli2013}
Roncalli T (2013) Introducing expected returns into risk parity portfolios: A
  new framework for tactical and strategic asset allocation. Available at SSRN:
  \url{http://ssrncom/abstract=2321309}

\bibitem[{Spinu(2013)}]{Spinu2013}
Spinu F (2013) An algorithm for computing risk parity weights. Available at
  SSRN: \url{http://ssrncom/abstract=2297383}

\bibitem[{Tasche(2002)}]{tasche2002expected}
Tasche D (2002) Expected shortfall and beyond. Journal of Banking \& Finance
  26(7):1519--1533

\bibitem[{Tawarmalani and Sahinidis(2005)}]{tawarmalani2005polyhedral}
Tawarmalani M, Sahinidis NV (2005) A polyhedral branch-and-cut approach to
  global optimization. Mathematical Programming 103(2):225--249

\bibitem[{Yang and Wei(2008)}]{GenEuler}
Yang F, Wei Z (2008) Generalized {E}uler identity for subdifferentials of
  homogeneous functions and applications. Journal of Mathematical Analysis and
  Applications 337:516--523

\end{thebibliography}
}

\end{document}